\documentclass[11pt]{article}

\usepackage{longtable, ltcaption}%
\usepackage{amsmath,amsthm,amssymb,bbm}
\usepackage[margin=.7in]{geometry}
\usepackage{hyperref}
\hypersetup{
    colorlinks=true,
    linkcolor=blue,
    citecolor=blue,
    urlcolor=blue
}
\usepackage{setspace,placeins,booktabs,ctable}
\usepackage{cleveref, graphicx, caption, subcaption}
\usepackage[utf8]{inputenc}
\usepackage{bm}
\usepackage{comment}
\excludecomment{additional-material}

\usepackage[bibstyle=numeric, citestyle=authoryear, doi=false, url=true, backend=biber, maxbibnames=10, maxcitenames=4, uniquelist=false, uniquename=false, sorting=nyt]{biblatex}
\renewbibmacro{in:}{}
\DeclareNameAlias{default}{last-first/first-last}
\newcommand{\citet}{\textcite}
\newcommand\independent{\protect\mathpalette{\protect\independenT}{\perp}}
    \def\independenT#1#2{\mathrel{\setbox0\hbox{$#1#2$}%
    \copy0\kern-\wd0\mkern4mu\box0}} 
\newcommand{\E}{\mathbbm{E}}
\renewcommand{\L}{\textrm{L}}
\renewcommand{\P}{\textrm{P}}

\newcommand{\point}[1]{}

\newtheorem{theorem}{Theorem}

\newtheorem{lemma}{Lemma}

\newtheorem{corollary}{Corollary}
\newtheorem{assumption}{Assumption}

\newtheorem{proposition}{Proposition}

\newtheorem{inneruassumption}{Assumption}
\newenvironment{namedassumption}[1]
  {\renewcommand\theinneruassumption{#1}\inneruassumption}
  {\endinneruassumption}

\crefname{figure}{Figure}{Figures}
\crefname{assumption}{Assumption}{Assumptions}
\crefname{inneruassumption}{Assumption}{Assumptions}

\theoremstyle{definition}

\newtheorem{example}{Example}
\newtheorem{remark}{Remark}

\date{February 8, 2022}

\title{\textbf{Difference in Differences with Time-Varying Covariates}}

\author{Carolina Caetano\footnote{Department of Economics, University of Georgia.  \href{mailto:carol.caetano@uga.edu}{carol.caetano@uga.edu}} \and Brantly Callaway\footnote{Department of Economics, University of Georgia.  \href{mailto:brantly.callaway@uga.edu}{brantly.callaway@uga.edu}} \and Stroud Payne\footnote{Department of Economics, University of Georgia. \href{mailto:robert.payne@uga.edu}{robert.payne@uga.edu}} \and Hugo Sant'Anna Rodrigues\footnote{Department of Economics, University of Georgia. \href{mailto:hsantanna@uga.edu}{hsantanna@uga.edu}}}

\begin{document}

\maketitle

\abstract{\noindent This paper considers identification and estimation of causal effect parameters from participating in a binary treatment in a difference in differences (DID) setup when the parallel trends assumption holds after conditioning on observed covariates.  Relative to existing work in the econometrics literature, we consider the case where the value of covariates can change over time and, potentially, where participating in the treatment can affect the covariates themselves.  We propose new empirical strategies in both cases.  We also consider two-way fixed effects (TWFE) regressions that include time-varying regressors, which is the most common way that DID identification strategies are implemented under conditional parallel trends.  We show that, even in the case with only two time periods, these TWFE regressions are not generally robust to (i) time-varying covariates being affected by the treatment, (ii) treatment effects and/or paths of untreated potential outcomes depending on the level of time-varying covariates in addition to only the change in the covariates over time, (iii) treatment effects and/or paths of untreated potential outcomes depending on time-invariant covariates, (iv) treatment effect heterogeneity with respect to observed covariates, and (v) violations of strong functional form assumptions, both for outcomes over time and the propensity score, that are unlikely to be plausible in most DID applications.  Thus, TWFE regressions can deliver misleading estimates of causal effect parameters in a number of empirically relevant cases.  We propose both doubly robust estimands and regression adjustment/imputation strategies that are robust to these issues while not being substantially more challenging to implement.}

\vspace{50pt}


\bigskip


\vspace{100pt}

\clearpage

\normalsize

\onehalfspacing

\section{Introduction}

In this paper, we study difference in differences identification strategies where (i) the parallel trends assumption holds only after conditioning on covariates, (ii) some or all of these covariates vary over time, and (iii) some of the time varying covariates could themselves be affected by the treatment.  

A number of papers (e.g., \citet{heckman-ichimura-smith-todd-1998,abadie-2005,santanna-zhao-2020}) show that certain causal effect parameters, typically the average treatment effect on the treated (ATT), are identified under conditional parallel trends assumptions.  These types of conditional parallel trends assumptions are attractive in applications where the path of untreated potential outcomes may differ among units with different characteristics.  However, work in the econometrics literature typically considers the case where covariates involved in the parallel trends assumption either do not vary over time or are ``pre-treatment'' (that is, the value of a time-varying covariate is set to its value in the pre-treatment period; see \citet{bonhomme-sauder-2011,lechner-2011} for some discussions on using pre-treatment values of time-varying covariates).  In contrast, empirical work in economics often only includes covariates that vary over time. In this case, identification must implicitly assume that the treatment does not have an effect on the covariates themselves, which is implausible in some applications.

Covariates that could have been affected by participating in the treatment are often referred to as ``post-treatment'' or as ``bad controls.''  The received wisdom seems to be that this type of covariate should not be included in empirical research.\footnote{For example, \citet{angrist-pischke-2008} discuss ``bad controls'' in the context of deciding whether or not to control for occupation when studying causal effects of graduating from college on earnings.  In that case, occupation is likely to be affected by attending college and, therefore, can make comparisons in earnings among those with the same occupation who graduated or did not graduate from college hard to interpret (even if college were randomly assigned).  \citet{angrist-pischke-2008} note that ``...we would do better to control only for variables that are not themselves caused by education.''  We return to a related example later in this section on the effect of job displacement on earnings where occupation is potentially affected by job displacement.  \point{double check this as I don't have the real version of MHE}}  However, we provide several examples below where it seems important to condition on the value of the covariate that \textit{would have occurred in the absence of the treatment}; in these cases, it would not generally be sufficient to just ``not include'' this sort of covariate.  We propose several different strategies for dealing with time-varying covariates that show up in the parallel trends assumption while also potentially being affected by the treatment.  

Difference in differences identification strategies are most often implemented using two-way fixed effects (TWFE) regressions.  The most common version of a TWFE regression that includes covariates is the following 
\begin{align} \label{eqn:twfe}
    Y_{it} = \theta_t + \eta_i + \alpha D_{it} + X_{it}'\beta + v_{it}
\end{align}
where $\theta_t$ is a time fixed effect, $\eta_i$ is individual-level unobserved heterogeneity (i.e., an individual fixed effect), $D_{it}$ is the treatment indicator, and $X_{it}$ are time varying covariates.   In the TWFE regression in \Cref{eqn:twfe}, $\alpha$ is the parameter of interest and it is often interpreted as ``the causal effect of the treatment'' or at least would be hoped to be a weighted average of underlying heterogeneous treatment effects. Being able to include covariates is one of the main attractions of using a TWFE regression to implement a DID design.  For example, \citet{angrist-pischke-2008} write: ``A second advantage of regression-DD is that it facilitates empirical work with regressors other than switched-on/switched off dummy variables.''\footnote{ \citet{angrist-pischke-2008} also briefly mention ``bad controls'' in the context of difference in differences (Section 5.2.1).}  TWFE regressions have come under much scrutiny in recent work in terms of how well they perform for implementing DID identification strategies.  In particular, TWFE regressions can perform very poorly in the presence of more than two time periods, variation in treatment timing across units, and treatment effect heterogeneity (particularly, treatment effect dynamics); see \citet{goodman-2021,chaisemartin-dhaultfoeuille-2020}. Although with only two time periods, TWFE regressions are known to be reliable under unconditional parallel trends, here we point out a number of problems with TWFE regressions for implementing DID identification strategies that rely on conditional parallel trends assumptions \textit{even in the case with only two time periods}.  

In particular, we show that TWFE regressions can deliver poor estimates of the average treatment effect on the treated (which is the natural target parameter for DID identification strategies) for any of four reasons: (1) time-varying covariates that are themselves affected by the treatment, (2) ATTs and/or parallel trends assumptions that depend on the pre-treatment level of time varying covariates in addition to (or instead of) only the change in the covariates over time, (3) ATTs and/or paths of untreated potential outcomes that depend on time-invariant covariates, and (4) violations of strong functional form assumptions both for outcomes over time and for the propensity score.  All four of these issues are common in applications in economics.  

In applications where none of the four issues mentioned above occur, TWFE regressions deliver a weighted average of conditional ATTs where all the weights are positive.  However, even in this best-case scenario, TWFE regressions still suffer from a ``weight-reversal'' property similar to the one pointed out in \citet{sloczynski-2020} under unconfoundedness with cross-sectional data.  In our case, conditional ATTs for relatively uncommon values of the covariates among the treated group (relative to the untreated group) are given large weights while conditional ATTs for common values of the covariates among the treated group are given small weights.  In order to get around this weight reversal issue, one needs to additionally rule out heterogeneous treatment effects across different values of the covariates.  Adding this condition to the previous four implies that TWFE regressions deliver the ATT; however, we stress that these are a very stringent set of requirements for TWFE regressions to perform well for estimating the ATT when the parallel trends assumption depends on time-varying covariates.  

We propose several new strategies for dealing with time varying covariates that are required for the parallel trends assumption to hold.  When the researcher is confident that the covariates evolve exogenously with respect to the treatment, we provide a doubly robust estimand for the ATT (these arguments are similar to the ones in \citet{santanna-zhao-2020} for the case with time invariant covariates).  Doubly robust estimators have the property that they deliver consistent estimates of the ATT if \textit{either} an outcome regression model or a propensity score model is correctly specified, thus giving researchers an extra chance to correctly specify a model relative to regression adjustment or propensity score weighting strategies.  Besides this, our doubly robust estimands can also be used in the context of the double/debiased machine learning literature where the propensity score and outcome regression model can be estimated using a wide variety of modern machine learning techniques (see \citet{chernozhukov-etal-2018} for the general case and \citet{chang-2020} in the context of DID).\footnote{Using machine learning in this context may be particularly useful because the expressions for the ATT involve conditioning on time-varying covariates across different time periods.  In many applications, time-varying covariates may be highly serially correlated, and it may be challenging to specify simple parametric models involving these covariates in this context.  However, machine learning estimators may perform much better in this context.} 
When the time-varying covariates can be affected by the treatment, we provide sufficient (and easy-to-interpret) conditions under which the strategy of conditioning on ``pre-treatment'' covariates, which is common in the econometrics literature, is justified.  We also discuss other cases where this strategy is not reasonable.  In these cases, we propose regression adjustment-type and doubly robust-type expressions for the ATT.  Finally, when a researcher is willing to make an additional function form assumption for untreated potential outcomes, we propose some even simpler approaches based on regression adjustment (these approaches are also broadly similar to recent ``imputation estimators'' proposed in \citet{liu-wang-xu-2021,gardner-2021,borusyak-jaravel-spiess-2021}).  We also show that stronger functional form assumptions for the model for untreated potential outcomes can allow for parallel trends-type assumptions for the covariates to be sufficient for identification of the ATT.


\bigskip

Before moving into our main arguments, we provide three examples to illustrate the types of questions that we address in the current paper.  We revisit these applications at relevant parts of the paper.

\begin{example}[Stand-your-ground laws] \label{ex:cheng-hoekstra} \citet{cheng-hoekstra-2013} study the effects of stand-your-ground laws on homicides and other crimes.  They use state-level data and exploit variation in the timing of stand-your-ground laws across states in order to identify policy effects.  For some of their results, they condition on time-varying covariates that include state-level demographics, the number of police officers in the state, the number of people incarcerated, median income, poverty rate, and spending on assistance and public welfare.  Although it is debatable whether or not some of the these covariates could be affected by the treatment (particularly the number of police officers and the number of people incarcerated), by running TWFE regressions that include these covariates, \citet{cheng-hoekstra-2013} at least implicitly argue that these covariates evolve exogenously from the treatment.  Whether this is true or not, for exposition purposes we will assume that none of the covariates used in this example are  affected by the treatment.  
\end{example}

\begin{example}[Shelter-in-place orders] \label{ex:covid} A number of recent papers study the effect of shelter-in-place orders on various outcomes including mobility (see, for example, \citet{weill-stigler-deschenes-springborn-2021} and references therein), labor market outcomes (e.g., \citet{gupta-montenovo-nguyen-rojas-schmutte-simon-weinburg-2020}, and consumer spending (e.g., \citet{chetty-friedman-hendren-stepner-2020}). \point{double check this last reference}  \point{none of these papers actually controls for the current number of cases; it would be good to have a paper that does this, but I think researchers are aware of this bad control issue and so are unlikely to include it}  Paths of all of these outcomes (in the absence of shelter-in-place orders) likely depend on the current number of Covid-19 cases due to individuals making different choices about staying at home or continuing to work based on the local ``state'' of the pandemic.  This suggests that parallel trends assumptions ought to condition on the number of Covid-19 cases that would have occurred if the policy had not been implemented.   Moreover, since Covid-related policies are \textit{designed} to affect the number of Covid-19 cases, this would be a case with a time-varying covariate that is likely to be affected by the treatment.

\end{example}

\begin{example}[Job Displacement]
Research on job displacement typically invokes parallel trends assumptions to identify causal effects of job displacement on workers' earnings.  If, in the absence of job displacement, paths of earnings depend on the occupation, industry, or union status of a worker, then it would be desirable to condition on these variables in the parallel trends assumption.  However, most empirical work on job displacement does not condition on these variables, presumably due to each of these possibly being affected by job displacement.\footnote{Some papers do include occupation, industry, and/or union controls in ``robustness checks'' and others study how effects of job displacement vary by whether or not a worker remains in the same industry, occupation, or union status following job displacement which is broadly similar to controlling for each of these (see, for example, \citet{topel-1991,jacobson-lalonde-sullivan-1993,stevens-1997}).}  
    Moreover, \citet{barnette-odongo-reynolds-2021} argue that differences in the distribution of pre-displacement occupations are likely an important explanation for the magnitude of effects of job displacement; similarly, \citet{brand-2006} reports relatively large effects of job displacement on occupation.
    \point{it would be good to have some paper that included occupation even as a robustness check}
\end{example}

The examples above are broadly representative of applications that invoke DID identification assumptions with time varying covariates.  The first example involves time-varying covariates that can reasonably be thought of as evolving exogenously with respect to the treatment.  The following two examples both involve covariates that are potentially affected by the treatment.  Later in the paper, we point out some further conceptual differences between these latter two examples.

\paragraph{Related Literature} \

Our paper shares a similar motivation to \citet{zeldow-hatfield-2021} which considers different possible sources of bias due to controlling for time-varying covariates that are possibly affected by the treatment.  That paper mainly considers how sensitive existing strategies are (e.g., controlling for only pre-treatment covariates or additionally including lagged outcomes) to covariates that can be affected by the treatment.  Relative to that paper, we make explicit assumptions on how the treatment can affect the covariates and, under these extra conditions, are able to propose estimation strategies that are guaranteed to perform well (up to regularity conditions) in those cases.

Our paper is also related to the literature on causal inference with panel data using structural nested mean models (\citet{robins-1997}) and marginal structural models (\citet{robins-hernan-brumback-2000}); see \citet{blackwell-glynn-2018} for a recent review.  These approaches, however, are based on ``sequential ignorability'' assumptions rather than allowing for time-invariant unobserved heterogeneity.  Sequential ignorability implies that treated and untreated potential outcomes are independent of treatment status conditional on pre-treatment values of covariates (and possibly pre-treatment outcomes).\footnote{Another difference between the current paper and much of the sequential ignorability literature is that these papers are typically primarily interested in recovering causal effects of different treatment paths (e.g., where each unit can move into or out of the treatment in each period).  The arguments in our paper could likely be extended in this direction but our main results apply to the case where there are only two time periods and treatment can only take place in the second time period.}  Unlike the bulk of this literature, the current paper focuses on the case where a researcher would like to invoke a parallel trends assumptions -- rather than sequential ignorability -- for identification.  However, the current paper also invokes an additional assumption on how treated and untreated potential covariates are generated; this type of assumption is not made in this literature.  The reason for this is that the timing that we consider differs from what is typically considered in the literature on sequential ignorability; in our case, units potentially become treated, then their covariate realizes (and may itself be affected by treatment) and this covariate needs to be controlled for identification.  By contrast, the sequential ignorability literature typically has the covariate realized first, then the treatment, then the outcome, and controlling for, effectively, the covariate in the previous period is sufficient for identifying parameters of interest.  That said, like the current paper, that literature does take seriously how covariates evolve over time and how participating in the treatment can affect covariates themselves.  Of papers broadly in this literature, the most similar to the current paper is \citet{imai-kim-wang-2018} which focuses on a conditional parallel trends assumption that can hold after conditioning on past values of the covariates as well as past values of the outcome.  

Our paper is also related to the literature on mediation analysis.  Like a mediator, our covariates can be affected by treatment participation.  However, the mediation literature is typically interested in decomposing treatment effects into direct effects of the treatment and indirect effects due to the effect of the treatment on the mediator (see \citet{huber-2020} for a recent review of this literature).  Our paper is less ambitious in that we only seek to identify the overall effect of the treatment on outcomes; the tradeoff is that we are able to generally make weaker assumptions than would be required to separately recover direct and indirect effects of participating in the treatment.  That said, it would be interesting to extend our arguments to additionally identifying direct and indirect effects of participating in the treatment, and it seems likely that existing arguments from the mediation analysis literature could be applied in this case.  Our paper is relatively more similar to \citet{rosenbaum-1984,lechner-2008,flores-lagunes-2009}; these papers consider identification of treatment effect parameters under unconfoundedness (and with cross-sectional data) where the covariates that are required for the unconfoundedness assumption to hold could have been affected by the treatment.  Besides this, our paper is related to a large literature in econometrics on strict exogeneity and pre-determinedness in panel data models (see, for example, \citet{arellano-honore-2001}).

Finally, our results on interpreting TWFE regressions build on work on interpreting cross-sectional regressions under the assumption of unconfoundedness and in the presence of treatment effect heterogeneity; this literature includes \citet{angrist-1998,aronow-samii-2016,sloczynski-2020,goldsmith-hull-kolesar-2021,ishimaru-2021}.  \citet{goodman-2021,chaisemartin-dhaultfoeuille-2020,ishimaru-2022} all provide decompositions of the TWFE regression in \Cref{eqn:twfe}.  In some ways, the decompositions in these papers are more general than our decomposition as they all consider the case with more than two time periods and with variation in treatment timing.  On the other hand, our results zoom in on the ``textbook'' case with exactly two periods and where no one is treated in the first period; our decomposition emphasizes a number of possible limitations of TWFE regressions even in the case with exactly two periods.  Indeed, moving to more complicated cases with more periods and variation in treatment timing would make the case for using TWFE regressions even weaker, as it would introduce additional issues particularly related to using already treated units as comparison units (which can lead to negative weights on underlying treatment effect parameters), as all three papers mentioned above imply.  See \Cref{rem:twfe-decomposition-comparison} below for a more detailed comparison.

\section{Identification} \label{sec:identification}

\subsection*{Notation and Setup}

For this section, we focus on a baseline case where the researcher has access to two time periods of panel data.  We label the second time period $t^*$ and the first time period $t^*-1$, and use $t$ to indicate a generic time period.  In each time period, we observe outcomes $Y_t$, a time-varying covariate $X_t$, and time invariant covariates $Z$.  As is standard in the DID literature, we suppose that no one is treated in the first time period.  We use the binary variable $D$ to indicate whether or not a unit participates in the treatment.  Importantly for our setup, we allow for the possibility that the time varying covariate can itself be affected by the treatment; in order to do this, we define $X_{t}(1)$ to be the value that the covariate would take if a unit participated in the treatment and $X_{t}(0)$ to be the value that the covariate would take if a unit did not participate in the treatment; for simplicity, we often refer to these as ``treated potential covariates'' and ``untreated potential covariates.''  Next, we define treated potential outcomes as $Y_t(1,X_t(1))$ (this is the outcome that a unit would experience in time period $t$ if they participated in the treatment and their covariate took on its value under the treatment) and untreated potential outcomes as $Y_t(0,X_t(0))$ (this is the outcome that a unit would experience in time period $t$ if they did not participate in the treatment and their covariate took its value in the absence in the treatment).  For most of the arguments in the current paper, it is sufficient to use the shorter notation $Y_t(1) := Y_t(1,X_t(1))$ and $Y_t(0) := Y_t(0,X_t(0))$.  In this setup, the observed covariates in each time period are: $X_{t^*} = D X_{t^*}(1) + (1-D)X_{t^*}(0)$ and $X_{t^*-1} = X_{t^*-1}(0)$.  In other words, in the second time period, we observe treated potential covariates for units that participate in the treatment, and we observe untreated potential covariates for units that do not participate in the treatment.  In the first time period, since no units are treated yet, we observe untreated potential covariates for all units.  Likewise, observed outcomes are given by $Y_{t^*} = DY_{t^*}(1) + (1-D)Y_{t^*}(0)$, and $Y_{t^*-1} = Y_{t^*-1}(0)$.

Following the vast majority of the DID literature, we target identifying the average treatment effect on the treated (ATT).  It is given by
\begin{align*}
    ATT = \E[Y_{t^*}(1) - Y_{t^*}(0) | D=1]
\end{align*}
which is the average difference between treated and untreated potential outcomes among the treated group.

Throughout the paper, we make the following assumptions

\begin{assumption}[Random Sampling] \label{ass:sampling} The observed data consists of \\ $\{Y_{it^*}, Y_{it^*-1}, X_{it^*}, X_{it^*-1}, W_{it^*}, W_{it^*-1}, Z_i, D_i \}_{i=1}^n$ which are independent and identically distributed.
\end{assumption}

\begin{assumption}[Conditional Parallel Trends]  \label{ass:conditional-parallel-trends}
\begin{align*}
    \E[\Delta Y_{t^*}(0) | X_{t^*}(0), X_{t^*-1}, Z, D=1] = \E[\Delta Y_{t^*}(0) | X_{t^*}(0), X_{t^*-1}, Z, D=0] \ a.s.
\end{align*}
\end{assumption}

\begin{assumption}[Overlap] \label{ass:overlap} \ 
\begin{itemize}
    \item [(a)] $\P(D=1|X_{t^*},X_{t^*-1},Z) < 1$ a.s. 
    \item [(b)] $\P(D=1|X_{t^*-1},W_{t^*-1},Z) < 1$ a.s. 
\end{itemize}\point{double check}

\end{assumption}

\Cref{ass:sampling} says that we observe iid panel data.  \Cref{ass:conditional-parallel-trends} says that, on average, the path of untreated potential outcomes is the same for the treated group as for the untreated group after conditioning on untreated potential covariates in time period $t^*,$ pre-treatment covariates $X_{t^*-1}$, and time-invariant covariates $Z$.  Relative to standard conditional parallel trends assumptions (\citet{heckman-ichimura-todd-1997,abadie-2005,callaway-santanna-2021}), the set of covariates being conditioned on includes untreated potential covariates which are unobserved for the treated group and therefore can complicate existing identification strategies.  \Cref{ass:overlap} is an overlap assumption, and this typo of assumption is standard in the treatment effects literature.  Part (a) implies that, for any values of $X_{t^*}$, $X_{t^*-1}$, and $Z$, there will be some untreated units with those values of the covariates in the population.  Part (b) is similar but holds for any values of $X_{t^*-1}$, $W_{t^*-1}$, and $Z$. \point{I don't think that we need extra conditions for the arguments in Part (2) of Theorem 1, but be careful here.}

Next, we provide two distinct assumptions for dealing with covariates that vary over time.

\begin{namedassumption}{Cov-Exogeneity} \label{ass:cov-exog} $(X_{t^*}(0) | X_{t^*-1}, Z, D=1) \sim (X_{t^*}(1) | X_{t^*-1}, Z, D=1)$
\end{namedassumption}

\begin{namedassumption}{Cov-Unconfoundedness} \label{ass:cov-unc}
$X_{t^*}(0) \independent D | X_{t^*-1}, W_{t^*-1}, Z$ where $W_{t^*-1}$ is a vector of pre-treatment variables. 
\end{namedassumption}

We call the first assumption covariate exogeneity because it implies that participating in the treatment does not change the distribution of covariates for the treated group.  This assumption is technically weaker than assumptions like, for all units $X_{it^*}(1) = X_{it^*}(0) = X_{it^*}$ though this would certainly be a leading case where this sort of condition might hold.  \Cref{ass:cov-exog} allows for covariates to change values over time, but it imposes that (in distribution) they are not affected by participating in the treatment.  This sort of condition may be reasonable in some applications (e.g., Example 1 above).  In other cases, this assumption may be less reasonable (e.g., Examples 2 and 3 above).

\Cref{ass:cov-unc} is an unconfoundedness assumption for untreated potential covariates.  It allows for the treatment to effect the time varying covariates, but it says that the distribution of untreated potential covariates is the same for the treated group and the untreated group after conditioning on the vector of pre-treatment covariates $(X_{t^*-1}, W_{t^*-1}, Z)$.  This assumption allows us to recover the conditional distribution of untreated potential covariates for the treated group.  This distribution is a key ingredient for identifying the ATT below.

In \Cref{ass:cov-unc}, we allow for the possibility that $W_{t^*-1}$ is empty; in fact, this is a leading case.  In this case, unconfoundedness for untreated potential covariates holds after conditioning on the lag of the time-varying covariates $X_{t^*-1}$ and other time invariant covariates $Z$.  Below, we connect this specific condition to the common practice in the econometrics literature on DID of conditioning on pre-treatment values of time-varying covariates.  With a slight abuse of notation, we also allow for the possibility that $W_{t^*-1}$ includes the lagged outcome $Y_{t^*-1}$.  For example, another interesting case is when $W_{t^*-1} = Y_{t^*-1}$, so that covariate unconfoundedness holds after conditioning on pre-treatment covariates, time invariant covariates, and the pre-treatment outcome.  Interestingly, we show below that, under this condition, both the path of outcomes over time and the lag of the outcome show up in the expression for $ATT$ which is unusual in DID applications (see, \citet{chabe-2017} for related discussion).  In the results below, we provide separate results that invoke either \Cref{ass:cov-exog} or \Cref{ass:cov-unc}.

Next, we state our main identification result.

\begin{theorem} \label{thm:1} Under \Cref{ass:sampling,ass:conditional-parallel-trends},

\begin{itemize}
    \item[(1)] if, in addition, \Cref{ass:cov-exog} and \Cref{ass:overlap}(a) hold, then
    \begin{align*}
        ATT = \E[\Delta Y_{t^*} | D=1] - \E\Big[ \E[\Delta Y_{t^*} | X_{t^*}, X_{t^*-1}, Z, D=0] \big| D=1 \Big].
    \end{align*}
    \item[(2)] if, in addition, \Cref{ass:cov-unc} and \Cref{ass:overlap}(b) hold, then
    \begin{align*}
        \hspace{-.4cm}ATT = \E[\Delta Y_{t^*} | D=1] - \E\Big[ \E\big[ \E[\Delta Y_{t^*} | X_{t^*}, X_{t^*-1}, Z, D=0] \big| X_{t^*-1}, W_{t^*-1}, Z, D=0 \big] \Big| D=1 \Big].
    \end{align*}
\end{itemize}
\end{theorem}

The intuition for part (1) of \Cref{thm:1} is relatively straightforward.  Under the conditional parallel trends assumption and when covariates evolve exogenously, one can recover the ATT by (i) taking the path of outcomes experienced by the treated group and adjusting it by the path of outcomes experienced by the untreated group (conditional on $X_{t^*}$, $X_{t^*-1}$, and $Z$) and then (ii) accounting for differences in the distribution of $X_{t^*}$, $X_{t^*-1}$, and $Z$ across groups.  This result is very similar to existing results with time invariant covariates (e.g., \citet{heckman-ichimura-todd-1997}) as well as \citet{lechner-2011}).

The intuition for part (2) is somewhat more complicated.  The term $\E[\Delta Y_{t^*}|X_{t^*},X_{t^*-1},Z,D=0]$ is the average change in outcomes over time conditional on $X_{t^*}$, $X_{t^*-1}$, and $Z$ among the untreated group.  Under \Cref{ass:conditional-parallel-trends}, this is the path of outcomes that, conditional on $X_{t^*}(0), X_{t^*-1}$, and $Z$, the treated group would have experienced if they had not participated in the treatment.  The next expectation is over the distribution of $X_{t^*}(0)$ (conditional on $X_{t^*-1}, W_{t^*-1}$, and $Z$) for the untreated group.  Under \Cref{ass:cov-unc}, this is the same conditional distribution that $X_{t^*}(0)$ follows for the treated group.  Finally, the outside expectation is over the distribution of $X_{t^*-1}$, $W_{t^*-1}$, and $Z$ for the treated group and, therefore, allows for these variables to be distributed differently in the treated group relative to the untreated group.

\begin{corollary}[Important Special Cases] \label{cor:1} Under Assumptions \ref{ass:sampling}, \ref{ass:conditional-parallel-trends}, \ref{ass:overlap}(b), and \ref{ass:cov-unc},
\begin{itemize}
    \item[(1)] if, in addition, $W_{t^*-1} = \varnothing$, then
    \begin{align*}
        ATT = \E[\Delta Y_{t^*} | D=1] - \E\Big[ \E[\Delta Y_{t^*} | X_{t^*-1}, Z, D=0 ] \Big| D=1\Big].
    \end{align*}
    \item[(2)] if, in addition, $W_{t^*-1} = Y_{t^*-1}$, then 
    \begin{align*}
        \hspace{-.3cm}ATT = \E[\Delta Y_{t^*} | D=1] - \E\Big[ \E\big[ \E[\Delta Y_{t^*} | X_{t^*}, X_{t^*-1}, Z, D=0] \big| X_{t^*-1}, Y_{t^*-1}, Z, D=0 \big] \Big| D=1 \Big].
    \end{align*}
\end{itemize}
\end{corollary}

\Cref{cor:1} provides two important special cases for the results in part (2) of \Cref{thm:1}.  The first part provides a formal justification for the common practice in the econometrics literature on DID with time varying covariates of including only ``pre-treatment'' covariates.  In particular, this result says that, when unconfoundedness holds for the time varying covariate conditional on time-invariant covariates and other pre-treatment covariates, then it is sufficient for the researcher to only ``account for'' pre-treatment and time-invariant covariates in order to recover the ATT.\footnote{We provide the proof of this result in \Cref{app:proofs}.  The proof is relatively straightforward, but it appears to be a new contribution in the literature.}  The second part of \Cref{cor:2} is also interesting in that it relates the ATT to an expression that includes the lagged outcome.  There are a number of papers that explore the idea of including lagged outcomes in a DID framework (e.g., \citet{chabe-2017,imai-kim-wang-2018,zeldow-hatfield-2021}) though it is challenging to provide a justification for including lagged outcomes in DID settings --- our approach justifies the  inclusion of lagged outcomes (in the manner specified in the corollary) in cases where unconfoundedness for the time-varying covariate holds after conditioning on the lag of the outcome variable.

Next, we provide alternative expressions for $ATT$ that are useful for estimation.  

\begin{corollary}[Doubly Robust Expressions for ATT] \label{cor:2}
    Under \Cref{ass:sampling,ass:conditional-parallel-trends}, 
    \begin{itemize}
    \item[(1)] if, in addition, \Cref{ass:cov-exog} and \Cref{ass:overlap}(a) hold, then
    \begin{align} \label{eqn:dr-att-cov-exog}
       \hspace{-1cm} ATT = \E\left[ \left( \frac{D}{\E[D]} - \frac{p(X_{t^*},X_{t^*-1},Z) (1-D)}{\E[D](1-p(X_{t^*},X_{t^*-1},Z))} \right) \left( \Delta Y_{t^*} - \E[\Delta Y_{t^*} | X_{t^*}, X_{t^*-1}, Z, D=0]\right) \right],
    \end{align}
    where $p(X_{t^*},X_{t^*-1},Z) := \P(D=1|X_{t^*},X_{t^*-1},Z)$.
    \item[(2)] if, in addition, \Cref{ass:cov-unc} and \Cref{ass:overlap}(b) hold with $W_{t^*-1} = \varnothing$, then 
    \begin{align*}
        ATT = \E\left[ \left( \frac{D}{\E[D]} - \frac{p(X_{t^*-1},Z) (1-D)}{\E[D](1-p(X_{t^*-1},Z))} \right) \left( \Delta Y_{t^*} - \E[\Delta Y_{t^*} | X_{t^*-1},Z, D=0]\right) \right],
    \end{align*}
    where $p(X_{t^*-1},Z) := \P(D=1|X_{t^*-1},Z)$.
\end{itemize}
\end{corollary}

Both of the expressions in \Cref{cor:2} involve both an outcome regression (the conditional expectation terms in each expression) and a propensity score.  They are both doubly robust in the sense that a researcher can consistently estimate the ATT if \textit{either} the model for the propensity score or the outcome regression model is correctly specified (references on double robustness include  \citet{robins-rotnitzky-zhao-1994,scharfstein-rotnitzky-robins-1999,sloczynski-wooldridge-2018,santanna-zhao-2020}).  Besides this, they also provide a connection to the DID literature on estimating the ATT under conditional parallel trends using double/debiased machine learning; see, in particular, \citet{chang-2020}.  This may be particularly useful in the first case where the propensity score and outcome regression depends on time-varying covariates in both periods.  These can be practically difficult to estimate because, in many cases, $X_{t^*}$ and $X_{t^*-1}$ may be highly collinear.  Conventional methods typically invoke functional form assumptions that impose, for example, that these functionals only depending on $\Delta X_{t^*}.$ As noted below, these sorts of restrictions may be implausible in many applications.

\begin{remark}
  In cases where time-varying covariates may be affected by the treatment, we mainly focus on the case where an unconfoundedness type assumption holds for the time varying covariates.  A natural alternative would be to invoke parallel trends assumptions for the time-varying covariates themselves.  Importantly, though, our above arguments require identifying the entire conditional distribution of $X_{t^*}(0)$ for the treated group (not just its mean).\footnote{In \Cref{sec:imputation}, we propose some alternative approaches where standard parallel trends assumptions for time-varying covariates can be used though these approaches require imposing a linear model for the path of untreated potential outcomes in \Cref{ass:conditional-parallel-trends} that are not used in this section.}   That said, difference in differences approaches that recover the distribution of untreated potential outcomes, such as \citet{callaway-li-2019,callaway-li-oka-2018}, could be applied here (though note that these approaches require additional assumptions).  Likewise, the change-in-changes approach in \citet{athey-imbens-2006,melly-santangelo-2015}, which can recover distributions of untreated potential outcomes, could be applied to the time-varying covariates in this context.  Another potential limitation of these approaches in this context is that they typically only point identify distributions of continuous outcomes and, therefore, would not be very suitable for a number of relevant applications that involve discrete time-varying covariates.
\end{remark}

\begin{remark} Although neither of our assumptions on untreated potential covariates in \Cref{ass:cov-exog,ass:cov-unc} are directly testable, the condition in \Cref{ass:cov-unc} can be ``pre-tested'' --- that is, one can check if it holds in pre-treatment time periods.  One simple idea is to compute pseudo-ATTs in pre-treatment periods; if both \Cref{ass:conditional-parallel-trends} (the conditional parallel trends assumption) and \Cref{ass:cov-unc} hold in pre-treatment periods, then these pseudo-ATTs should be equal to 0.  Alternatively, one can directly pre-test \Cref{ass:cov-unc}: for some pre-treatment period $t$, \Cref{ass:cov-unc} implies that the distribution of $X_t | X_{t-1}, W_{t-1}, Z, D=d$ is the same for both the treated and untreated groups.  This sort of test could be implemented using results from the goodness-of-fit testing literature (e.g., \citet{bierens-1982,stute-1997}).
\end{remark}

\bigskip

To conclude this section, we revisit the three examples from the introduction.

\paragraph{Example 1 (Stand-your-ground, cont'd)} Our example on stand-your-ground laws involved conditioning on time-varying covariates that evolved exogenously with respect to the treatment.  This suggests that this example is most related to our results in part (1) of \Cref{thm:1} and part (1) of \Cref{cor:2}.  In particular, machine learning estimators of the propensity score and outcome regression functions in \Cref{cor:2} are particularly attractive as they do not require strong functional form assumptions on these nuisance functions.\footnote{This particular application uses state-level data, so, in practice, it may be difficult to use machine learning approaches with only 50 or so observations.  See \Cref{sec:imputation} for some more parametric approaches that may be more suitable for applications with limited data.  That said, the more general point here though is that, in cases where covariates evolve exogenously, machine learning estimators, given enough data, are likely to be attractive in many applications.}

\paragraph{Example 2 (Shelter-in-place, cont'd)}  In our example of shelter-in-place orders on various economic outcomes, the parallel trends assumption held after conditioning on the number of Covid-19 cases that would have occurred if the policy had not been implemented.  That is, ``untreated potential Covid-19 cases'' plays the role of $X_{t^*}(0)$ in this case.  \citet{callaway-li-2021b} show that, under a SIRD model --- which is the leading pandemic model in the epidemiology literature --- controlling for the pre-treatment ``state'' of the pandemic is sufficient for unconfoundedness to hold.  That is, the conditions 
in part (1) of \Cref{cor:1} and part (2) of \Cref{cor:2} hold when one wants to control for the number of untreated potential Covid-19 cases.  

\paragraph{Example 3 (Job displacement, cont'd)} Finally, recall our example on the effect of job displacement on earnings where the parallel trends assumption holds only after conditioning on, for example, ``untreated potential occupation'' --- that is, the occupation that a worker would have had if they had not been displaced from their job.  In this case, an unconfoundedness assumption for occupation may be more likely to hold if it conditions on (i) pre-treatment time-varying covariates (including pre-treatment occupation), (ii) time invariant covariates (such as demographics and education), \textit{and} (iii) pre-treatment earnings.  In particular, conditioning on pre-treatment earnings could be important if there are occupation specific wage premiums and high-earning workers are more likely to (in the absence of job displacement) stay in the same occupation over time relative to low-earning workers.  This application would then be covered by the results from part (2) of \Cref{cor:1}.

\section{Interpreting TWFE Regressions} \label{sec:twfe}

In this section, we consider how to interpret $\alpha$ in the TWFE regression in \Cref{eqn:twfe}.  We continue to consider the ``textbook'' case with two time periods where no one is treated in the first time period and where some (but not all) units become treated in the second time period.  This is a best-case for TWFE regressions as it does not introduce well-known problems related to using already-treated units as comparison units that show up when using TWFE regressions with multiple periods, variation in treatment timing, and treatment effect heterogeneity (\citet{goodman-2021,chaisemartin-dhaultfoeuille-2020}).  In the case with exactly two periods, it is helpful to equivalently re-write \Cref{eqn:twfe} as 
\begin{align} \label{eqn:fd}
  \Delta Y_{it^*} = \alpha D_i + \Delta X_{it^*}'\beta + \Delta v_{it^*}
\end{align}
where we define $\Delta X_{t^*} := (1,X_{t^*}-X_{t^*-1})'$ which is the change in the covariate over time and is augmented with an intercept term for the time fixed effect.  We also slightly abuse notation by taking $\beta$ to include an extra parameter in its first position corresponding to the intercept.  Our interest in this section is in determining what kind of conditions are required to interpret $\alpha$ as the ATT or at least as a weighted average of some underlying treatment effect parameters.  

Denote the linear projection of $\Delta Y_{t^*}$ on $\Delta X_{t^*}$ by $\L(\Delta Y_{t^*} | \Delta X_{t^*}) := \Delta X_{t^*}'\E[\Delta X_{t^*} \Delta X_{t^*}']^{-1} \E[\Delta X_{t^*} \Delta Y_{t^*}]$, and define the corresponding projection error $e := \Delta Y_{t^*} - \L(\Delta Y_{t^*} | \Delta X_{t^*})$.  Similarly, define the linear projection of $D$ on $\Delta X_{t^*}$ as $\L(D|\Delta X_{t^*}) := \Delta X_{t^*}'\E[\Delta X_{t^*} \Delta X_{t^*}']^{-1} \E[\Delta X_{t^*} D]$ and the corresponding projection error $u:= D - \L(D|\Delta X_{t^*})$.  Standard Frisch-Waugh type arguments imply that
\begin{align} \label{eqn:alpha}
    \alpha = \frac{\E[De]}{\E[u^2]}
\end{align}

Below, to keep the notation concise, it is useful to define $X^{all}(d) := (X_{t^*}(d), X_{t^*-1}, Z)$.  We also define $ATT_{X^{all}(0)}(X^{all}(0)) := \E[Y_{t^*}(1) - Y_{t^*}(0) | X^{all}(0), D=1]$ which is the ATT conditional on $X_{t^*}(0)$, $X_{t^*-1}$, and $Z$.  And we further define $p(X^{all}(0)) = \P(D=1|X^{all}(0))$.  Next, we state a main result decomposing $\alpha$ from the TWFE regression.

\begin{proposition} \label{prop:twfe} Under Assumptions \ref{ass:sampling}, \ref{ass:conditional-parallel-trends}, and \ref{ass:overlap}(a), 
\begin{align*}
    \alpha &= \E[\omega_{ATT}(X^{all}(0)) ATT_{X^{all}(0)}(X^{all}(0)) | D=1] \\
    & + \E\Bigg[ \sum_{d \in \{0,1\}} \omega_d(X^{all}(0)) \Bigg\{ \Big( \E[\Delta Y_{t^*} | X_{t^*}(0), X_{t^*-1}, Z, D=d] - \E[\Delta Y_{t^*} | X_{t^*}, X_{t^*-1}, Z, D=d] \Big) \tag{A}\\
    & \hspace{100pt} + \Big(\E[\Delta Y_{t^*} | X_{t^*}, X_{t^*-1}, Z, D=d] - \E[\Delta Y_{t^*} | X_{t^*}, X_{t^*-1}, D=d] \Big) \tag{B}\\
    & \hspace{100pt} + \Big(\E[\Delta Y_{t^*} | X_{t^*}, X_{t^*-1}, D=d] - \E[\Delta Y_{t^*} | \Delta X_{t^*}, D=d] ) \tag{C} \\
    & \hspace{100pt} + \Big( \E[\Delta Y_{t^*} | \Delta X_{t^*}, D=d] - \L(\Delta Y_{t^*} | \Delta X_{t^*}, D=d) \Big) \Bigg\} \Big| D=1 \Bigg] \tag{D} \\
    & + \E\left[ \omega_e(X^{all}(0)) \left\{ \L(\Delta Y_{t^*} | \Delta X_{t^*}, D=1) - \L(\Delta Y_{t^*} | \Delta X_{t^*}, D=0)\right\} \big| D=1\right] \tag{E}
\end{align*}
where 
\begin{align*}
    \omega_{ATT}(X^{all}(0)) &:= \frac{1-p(X^{all}(0))}{\E[(1-\L(D|\Delta X_{t^*}))|D=1]} \\
    \omega_d(X^{all}(0)) &:= \frac{d\, p(X^{all}(0)) + (1-d)(1-p(X^{all}(0)))}{\E[(1-\L(D|\Delta X_{t^*}))|D=1]} \\
    \omega_e(X^{all}(0)) &:= \frac{ (p(X^{all}(0)) - \L(D|\Delta X_{t^*}))}{\E[(1-\L(D|\Delta X_{t^*}))|D=1]}
\end{align*}
\end{proposition}

The result in \Cref{prop:twfe} indicates that $\alpha$ is equal to a weighted average of underlying conditional ATTs (we discuss the nature of the weights in more detail below) plus a number of undesirable ``bias'' terms.  We provide formal conditions under which each of these extra terms are equal to zero below.  But, informally, term (A) contains bias from the treatment potentially affecting the covariates in time period $t^*$.  Term (B) comes from ignoring time invariant covariates.  Term (C) comes up when paths of outcomes depend on the levels of time-varying covariates instead of only on the change in covariates over time.  Term (D) arises when the conditional expectation is nonlinear in the change in covariates over time.  Term (E) is conceptually different from terms (A)-(D) and is non-zero when the propensity score is not equal to a linear projection of the treatment on the change in covariates over time.\footnote{Without further assumptions, some of the expressions that involve $X_{t^*}(0)$ are not necessarily identified (this includes $ATT_{X^{all}(0)}(X^{all}(0)), \E[\Delta Y_{t^*}|X_{t^*}(0), X_{t^*-1},Z,D=1]$ and all of the weights as they depend on $p(X^{all}(0))$.  However, if we additionally invoke \Cref{ass:cov-exog}, then all of these terms are identified and Term (A) is equal to zero (see the discussion below for more details along these lines).  The reason we do not invoke this assumption in \Cref{prop:twfe} is to point out that time-varying covariates being affected by the treatment can itself be an additional complication for TWFE regressions.}

Next, we introduce several additional assumptions that are useful for eliminating the bias terms in \Cref{prop:twfe}.  We also use the additional notation:  $ATT_{X_{t^*}(0),X_{t^*-1}}(x_{t^*}(0), x_{t^*-1}) := \E[Y_{t^*}(1) - Y_{t^*}(0) | X_{t^*}(0)=x_{t^*}(0), X_{t^*-1}=x_{t^*-1}]$ and $ATT_{\Delta X_{t^*}(0)}(\Delta x_{t^*}(0)) := \E[Y_{t^*}(1) - Y_{t^*}(0) | \Delta X_{t^*}(0) = \Delta x_{t^*}(0)]$ --- these define different types of conditional ATTs.
\begin{assumption}[Conditional ATTs and parallel trends do not depend on time invariant covariates] \label{ass:bias-z} \ 

    \begin{itemize}
        \item[(a)] 
        $ATT_{X^{all}(0)}(X^{all}(0)) = ATT_{X_{t^*}(0), X_{t^*-1}}(X_{t^*}(0), X_{t^*-1})$ a.s.
        \item[(b)] $\E[\Delta Y_{t^*}(0) | X_{t^*}(0), X_{t^*-1}, Z, D=0] = \E[\Delta Y_{t^*}(0) | X_{t^*}(0), X_{t^*-1}, D=0]$ a.s.
    \end{itemize}
\end{assumption}

\begin{assumption}[Conditional ATTs and parallel trends only depend on change in time-varying covariates] \label{ass:bias-delta} \ 

    \begin{itemize}
        \item[(a)] $ATT_{X_{t^*}(0), X_{t^*-1}}(X_{t^*}(0), X_{t^*-1}) = ATT_{\Delta X_{t^*}(0)}(\Delta X_{t^*}(0))$ a.s.
        \item[(b)] $\E[\Delta Y_{t^*}(0) | X_{t^*}(0), X_{t^*-1}, D=0] = \E[\Delta Y_{t^*}(0) | \Delta X_{t^*}(0), D=0]$ a.s.
    \end{itemize}
\end{assumption}

\begin{assumption}[Linearity of conditional ATTs and paths of untreated potential outcomes] \label{ass:bias-linearity} \ 
\begin{itemize}
    \item [(a)] There exists a $\delta_1$ such that $ATT_{\Delta X_{t^*}(0)}(\Delta X_{t^*}(0)) = \Delta X_{t^*}(0)'\delta_1$
    \item [(b)] There exists a $\delta_0$ such that $\E[\Delta Y_{t^*}(0) | \Delta X_{t^*}(0), D=0] = \Delta X_{t^*}(0)'\delta_0$
\end{itemize}
\end{assumption}

\begin{assumption}[Linearity of propensity score in terms of change in time-varying covariates] \label{ass:bias-pscore} \ 

There exists a $\delta_p$ such that $p(X^{all}(0)) = \Delta X_{t^*}(0)'\delta_p$.
    
\end{assumption}

The first part of \Cref{ass:bias-z} says that, conditional on $X_{t^*}(0)$ and $X_{t^*-1}$, conditional ATTs do not depend on time invariant covariates $Z$.  The second part says that, conditional on $X_{t^*}(0)$ and $X_{t^*-1}$, the path of untreated potential outcomes does not depend on time invariant covariates $Z$.  This implies that the conditional parallel trends assumption holds without conditioning on time invariant covariates (and thus strengthens \Cref{ass:conditional-parallel-trends}).  \Cref{ass:bias-delta} is similar; the first part says that conditional ATTs further only depend on changes in time-varying covariates over time, and the second part says that the conditional parallel trends assumption only depends on the change in time-varying covariates over time rather than their level.

\Cref{ass:bias-linearity} says that conditional ATTs and paths of untreated potential outcomes are linear in changes in untreated potential covariates over time.  Jointly, \Cref{ass:cov-exog,ass:bias-z,ass:bias-delta,ass:bias-linearity} imply that (i) time varying covariates are not affected by the treatment, (ii) that conditional ATTs (conditional on $X_{t^*}(0), X_{t^*-1},$ and $Z$) only depend on the change in time-varying covariates (and not on their levels or time invariant covariates) and are linear in time-varying covariates, and (iii) that the conditional parallel trends assumption in \Cref{ass:conditional-parallel-trends} only depends on the change in time-varying covariates over time (and not on their levels or time invariant covariates) and is linear in time-varying covariates over time.

\Cref{ass:bias-pscore} says that the propensity score (conditional on $X_{t^*}(0)$, $X_{t^*-1}$, and $Z$) is linear in $\Delta X_{t^*}(0)$.  This type of assumption is very common in the literature on interpreting regressions under unconfoundedness with cross-sectional data (e.g., \citet{angrist-1998,aronow-samii-2016,sloczynski-2020,ishimaru-2021}).  In those cases, it sometimes holds by construction (e.g., when the covariates are all discrete and a full set of interactions is included in the model).  In our case, though, it seems particularly implausible as (i) it requires the propensity score to only depend on changes in covariates over time, and (ii) even with fully interacted discrete regressors, the propensity score is unlikely to be linear in changes in the regressors over time.\footnote{For example, suppose that the only covariate is binary.  In the cross-sectional case considered by other papers mentioned above, the propensity score would be linear by construction.  However, the change in the covariate over time would be a single variable that can take the values -1, 0, or 1; moreover, the change in a binary covariate over time is equal to 0 in cases when the covariate is equal to 1 in both periods or when the covariate is equal to 0 in both periods.  This suggests that the propensity score would not be linear in the change in covariates over time even in this very simple case. \point{make sure Carol thinks this is a good example}} 

\begin{proposition} \label{prop:twfe-extra-assumptions} Under Assumptions \ref{ass:sampling}, \ref{ass:conditional-parallel-trends}, \ref{ass:overlap}(a), \ref{ass:cov-exog}, \ref{ass:bias-z}, \ref{ass:bias-delta}, and \ref{ass:bias-linearity}, 
\begin{align*}
    \alpha &= \E[\omega_{ATT}(X^{all}(0)) ATT_{X^{all}(0)}(X^{all}(0)) | D=1] \\
    & \hspace{25pt} + \E\left[ \omega_e(X^{all}(0)) \left\{ \L(\Delta Y_{t^*} | \Delta X_{t^*}, D=1) - \L(\Delta Y_{t^*} | \Delta X_{t^*}, D=0)\right\} \big| D=1\right] \tag{E}
\end{align*}
where $\omega_{ATT}$ and $\omega_e$ are defined in \Cref{prop:twfe}.

\begin{itemize}
    \item [(a)] If, in addition, \Cref{ass:bias-pscore} holds, then
    \begin{align*}
    \alpha &= \E[\omega_{ATT}(X^{all}(0)) ATT_{X^{all}(0)}(X^{all}(0)) | D=1]
    \end{align*}
    and $\E[\omega_{ATT}(X^{all}(0)) | D=1] = 1$.

    \item[(b)] If, in addition, \Cref{ass:bias-pscore} holds and $ATT_{\Delta X_{t^*}(0)}(\Delta x_{t^*}(0))$ is the same across all values of $\Delta x_{t^*}(0)$, then
    \begin{align*}
        \alpha = ATT
    \end{align*}
\end{itemize}
\end{proposition}

The proof of \Cref{prop:twfe-extra-assumptions} is provided in \Cref{app:proofs}.  In the proof, we provide more specific results on which conditions are required for each term in terms (A)-(D) in \Cref{prop:twfe} to be equal to 0.

The result in \Cref{prop:twfe-extra-assumptions} suggests a number of potential issues with the TWFE regressions as in \Cref{eqn:twfe}.  First, even if one is willing to maintain the additional assumptions in \Cref{ass:cov-exog,ass:bias-z,ass:bias-delta,ass:bias-linearity} (which are likely to be very strong in most applications), $\alpha$ from the TWFE regression is still hard to interpret.  Maintaining these additional assumptions (particularly, \Cref{ass:cov-exog}) implies that all of the weights, conditional ATTs, and linear projections in the first part of \Cref{prop:twfe-extra-assumptions} are identified and directly estimable.  That said, the weights on conditional ATTs, $\omega_{ATT}$, do not have the property that they have mean one and the nuisance expression in term (E) may be non-negligible.  

The second part of \Cref{prop:twfe-extra-assumptions} says that, if we are willing to assume that the propensity score is equal to the linear projection of the treatment on the change in time-varying covariates over time, then the weights on conditional ATTs will have mean one and the nuisance expression in term (E) will be equal to zero.  Even in this case, the weights have a ``weight-reversal'' property analogous to the one pointed out in \citet{sloczynski-2020} in the context of unconfoundedness and cross-sectional data.  What this means is that conditional ATTs are given more weight for values of the covariates that are relatively uncommon among the treated group relative to the untreated group; and that conditional ATTs are given less weight for values of the covariates that are relatively common among the treated group relative to the untreated group.  

Finally, if in addition to all the previous conditions, conditional ATTs are constant across different values of the covariates, then $\alpha$ will be equal to the $ATT$.  This is a treatment effect homogeneity condition with respect to the covariates.  It \textit{is} somewhat weaker than individual-level treatment effect homogeneity and it allows for treatment effects to still be systematically different for treated units relative to untreated units; instead, for the treated group, treatment effects cannot be systematically different across different values of the covariates.

These results are much different from our earlier results in \Cref{sec:identification}.  Those results did not require any of the additional assumptions in \Cref{prop:twfe-extra-assumptions}.  In fact, when covariates evolve exogenously with respect to the treatment (as under \Cref{ass:cov-exog}), then the doubly robust expressions for the ATT in part (1) of \Cref{cor:2} only require that either the propensity score or the outcome regression model be correctly specified; in cases where these are estimated using machine learning, even these parametric assumptions can be substantially relaxed.  Moreover, in contrast with the TWFE regressions considered in this section, our earlier additional results can accommodate cases where the time-varying covariates are affected by the treatment.

\begin{remark}
  It is worth pointing out that all of the extra conditions considered in \Cref{prop:twfe-extra-assumptions} are sufficient conditions rather than necessary conditions.  For example, it would be possible for some violations of these assumptions to ``offset'' each other so that $\alpha$ \textit{happens} to be equal to ATT.  That said, there is no reason to expect this to happen in applications.
\end{remark}


\begin{remark} \label{rem:twfe-decomposition-comparison}
  The result in \Cref{prop:twfe} is related to several other decompositions of TWFE regressions that include covariates.  All of these papers consider the same TWFE regression that we do in \Cref{eqn:twfe}.  They also consider the more general case where there are more than two time periods and allow for variation in treatment timing.  \citet[Online Appendix 3.3]{chaisemartin-dhaultfoeuille-2020} show that, under a conditional parallel trends assumption that involves only changes in observed covariates and linearity assumptions that their main results related to multiple periods and variation in treatment timing essentially continue to apply.  That said, this suggests that in the two period case that we consider, TWFE regressions would recover the ATT.  The explanation for this difference is that we do not impose those extra conditions in \Cref{prop:twfe}.   \citet[Section 5.2]{goodman-2021} provides a decomposition of $\alpha$ into a ``within'' component and ``between'' component.  The between component arises due to variation in treatment timing, and, therefore does not show up in our case.  The within component is analogous to our expression for $\alpha$ in \Cref{eqn:alpha}.  Relative to \citet{goodman-2021}, we further decompose this term into a number of more primitive objects that highlight that researchers should be careful in interpreting ``within'' components as averages of causal effects unless they are willing to invoke extra assumptions.  Finally, \citet[Section 2.2]{ishimaru-2022}, like \citet{chaisemartin-dhaultfoeuille-2020}, provides conditions under which TWFE regressions that include covariates can be interpreted as weighted averages of underlying treatment effect parameters (though, in both papers, the weights can be negative).  These include a version of conditional parallel trends that holds when one conditions on the change in covariates over time\footnote{\citet{ishimaru-2022} does point out that ``conditioning on [changes in time-varying covariates] may not be sufficient to make parallel trends plausible.''} and an assumption on linearity of the propensity score conditional on changes in observed covariates over time.\footnote{Another way that the decomposition in \citet{ishimaru-2022} is more general than the one in the current paper is that paper does not require the treatment to be binary.  \citet{ishimaru-2022} also considers an interesting extension on decomposing a modified TWFE regression that additionally includes time-varying coefficients on time-varying coefficients.  Based on his result, it seems likely that this sort of regression would not suffer from issues related to parallel trends depending on the levels of time-varying covariates rather than only changes in time-varying covariates over time.  However, it appears that this regression would still suffer from the other issues mentioned in this section; that said, this is a distinct regression from the TWFE regression in \Cref{eqn:twfe} that is much more commonly used in empirical work in economics.}  None of these papers explicitly address the issue of time-varying covariates potentially being affected by the treatment.\footnote{\citet{goodman-2021} does remark that ``Note that for covariates to aid in identification, [time-varying covariates] must be unaffected by the treatment to avoid bias from `conditioning on a post-treatment variable'.''}
\end{remark}

\section{Alternative Regression Adjustment/Imputation Approaches} \label{sec:imputation}

In this section, we provide several alternative strategies that involve stronger parametric assumptions on the path of untreated potential outcomes than we made in \Cref{sec:identification}.  The approaches discussed in this section are generally simpler to estimate than would be the case for the expressions coming from  \Cref{sec:identification} and, in some cases, can allow for weaker (or at least alternative) assumptions on how the treatment affects time-varying covariates.  The strategies that we propose in this section are also able to avoid the issues with TWFE regressions pointed out in \Cref{sec:twfe}, and (when desired) can allow for the possibility that the treatment has an effect on the covariates.  

The ideas in this section are broadly similar to regression adjustment strategies in the treatment effects literature (see, for example, \citet[Section 5.3]{imbens-wooldridge-2009}) and the imputation estimators proposed in \citet{liu-wang-xu-2021,gardner-2021,borusyak-jaravel-spiess-2021} though they allow for (i) time-varying effects of time varying covariates and (ii) the possibility that the treatment directly affects the time-varying covariates.

To start with, it is well known (e.g. \citet{blundell-dias-2009,gardner-2021,borusyak-jaravel-spiess-2021}) that there is a close connection between unconditional parallel trends assumptions and the following model for untreated potential outcomes
\begin{align*}
    Y_{it}(0) = \theta_t + \eta_i + v_{it}
\end{align*}
where $\theta_t$ is a time fixed effect, $\eta_i$ is time invariant unobserved heterogeneity (i.e., an individual fixed effect), and $v_{it}$ are idiosyncratic, time varying unobservables.  An unconditional version of parallel trends holds in this model for untreated potential outcomes under the condition that $\E[\Delta v_{t} | D=1] = \E[\Delta v_t | D=0]$ for all time periods (this would hold by construction if $v_t$ is independent of treatment status in all time periods), but allows for $\eta$ to be distributed differently across groups and does not impose any modeling assumptions on treated potential outcomes.

As discussed above, the econometrics literature on difference in differences often considers the case where the covariates in the parallel trends assumption are time invariant.  In that case, the analogous model for untreated potential outcomes is given by

\begin{align*}
    Y_{it}(0) = g_t(Z_i) + \eta_i + v_{it}
\end{align*}
where the distribution of $\eta$ can vary across groups (as well as vary with $Z$) and the key condition for the conditional parallel trends assumption to hold is that $\E[\Delta v_t | Z, D=1] = \E[\Delta v_t | Z, D=0]$ (see, for example, \citet{heckman-ichimura-todd-1997} for a discussion of this kind of model).\footnote{To see this, notice that $\E[\Delta Y_t(0) | Z, D=1] = g_t(Z) - g_{t-1}(Z) = \E[\Delta Y_t(0) | Z, D=0]$ which implies that conditional parallel trends holds.}

In this setup, the main challenge is estimating $g_t(z)$ (though note that this is a practical, estimation challenge rather than an identification challenge).  The natural way to parameterize this model is
\begin{align} \label{eqn:linear-y0}
    Y_{it}(0) = Z_i'\delta_t + \eta_i + v_{it}
\end{align}
where we now take $Z$ to include an intercept (and, therefore, $\delta_t$ absorbs the time fixed effect).  Given this framework, $ATT=\E[\Delta Y_{t^*} | D=1] - \E[Z|D=1]'\delta_{t^*}^*$ where $\delta_{t^*}^* := (\delta_{t^*} - \delta_{t^*-1})$ which can be consistently estimated from the regression of $\Delta Y_{t^*}$ on $Z$ using only observations from the untreated group.\footnote{This is closely related to regression adjustment estimators (see, for example, \citet{heckman-ichimura-smith-todd-1998,imbens-wooldridge-2009,santanna-zhao-2020} for related discussion).   An alternative strategy would be to estimate the ATT by $n_1^{-1} \sum_{i=1,D_i=1}^n (Y_{it^*} - \hat{Y}_{it^*}(0))$ where $n_1$ is the number of treated observations and $\hat{Y}_{it^*}(0)$ is an imputed untreated potential outcome given by $\hat{Y}_{it^*}(0) = Y_{it^*-1} + Z_i'\hat{\delta}_{t^*}^*$.  This imputation estimator is numerically equal to the regression adjustment estimator, but the imputation formulation is convenient particularly in the case with multiple periods and variation in treatment timing; see \Cref{rem:multi-periods} below for more details.}

The same sort of arguments imply that, when there are some covariates that vary over time (as above, we consider the case of a single time-varying covariate but note that it is straightforward to extend these arguments to cases with more time-varying covariates), a natural motivating model is
\begin{align*}
    Y_{it}(0) = g_t(Z_i, X_{it}(0)) + \eta_i + v_{it}
\end{align*}
which implies that
\begin{align*}
    \Delta Y_{it^*}(0) = g_{t^*}(Z_i, X_{it^*}(0)) - g_{t^*-1}(Z_i,X_{it^*-1}) + \Delta v_{it^*}
\end{align*}
Moreover, the same sorts of arguments as above imply that \Cref{ass:conditional-parallel-trends} holds in this model.  Similar to the previous case, the main practical challenge is that $g_t(z,x_t(0))$ is likely to be challenging to estimate.  Like the previous case, the natural way to parameterize this model is
\begin{align*}
    Y_{it}(0) = Z_i'\delta_t + X_{it}(0)\beta_t + \eta_i + v_{it}
\end{align*}
which implies that
\begin{align} \label{eqn:imputation}
    \Delta Y_{it^*}(0) = Z_i'\delta_{t^*}^* + \Delta X_{it^*}(0)\beta_{t^*} + X_{it^*-1}(0) \beta_{t^*}^* + \Delta v_{it^*}
\end{align}
where $\beta_{t^*}^* := (\beta_{t^*} - \beta_{t^*-1})$.  Notice that, because untreated potential outcomes and untreated potential covariates are observed for the untreated group, the parameters in the model above can be recovered from a regression of the change in outcomes over time on time invariant covariates, the change in time varying covariates, and the level of the time varying covariates in the pre-treatment period.  The model in \Cref{eqn:imputation} is conceptually appealing as (up to the parametric assumptions) it compares units with both the same initial level of the time-varying covariate and that have the same change in time-varying covariates over time.

Although it is straightforward to recover the parameters in \Cref{eqn:imputation}, recall that, 
\begin{align*}
    ATT &= \E[\Delta Y_{t^*} | D=1] - \E[\Delta Y_{t^*}(0) | D=1] \\
    &= \E[\Delta Y_{t^*} | D=1] - \Big(\E[Z|D=1]'\delta_{t^*}^* + \E[\Delta X_{t^*}(0) | D=1] \beta_{t^*} + \E[X_{t^*-1}(0)|D=1]\beta_{t^*}^*\Big)
\end{align*}
Given that the parameters are identified, every term is identified in this expression except for $\E[\Delta X_{t^*}(0) | D=1]$ (because $X_{t^*}(0)$ is not observed for the treated group).  We briefly consider six settings for recovering $\E[\Delta X_{t^*}(0)|D=1]$ --- three of these come from the assumptions we have already considered for untreated potential covariates and three involve parallel trends assumptions for untreated potential covariates.  Several of these cases involve averaging over conditional expectations of $\Delta X_t(0)$. In this section we additionally impose linear models for these conditional expectations; under this extra condition, researchers are able to estimate ATT while potentially allowing for the treatment to affect time-varying covariates using only regressions and averaging.  

\paragraph{Case 1: \Cref{ass:cov-exog} holds} \ 

Under \Cref{ass:cov-exog}, $\E[\Delta X_{t^*}(0) | D=1] = \E[\Delta X_{t^*}(1) | D=1] = \E[\Delta X_{t^*} | D=1]$.  That is, when covariates evolve exogenously with respect to the treatment, we can replace the average change in untreated potential covariates for the treated group with the average change in covariates actually experienced by the treated group.

\paragraph{Case 2: \Cref{ass:cov-unc} holds conditional on $\bm{(Z,X_{t^*-1})}$} \ 

In this case, if we are willing to assume the following linear model for untreated potential covariates
\begin{align*}
    \Delta X_{it^*}(0) = Z_i'\gamma_{t^*} + X_{it^*-1}\lambda_{t^*} + u_{it^*}
\end{align*}
where it follows by the conditions in this case that $\E[u_{t^*}|Z,X_{t^*-1},D=d]=0$ for $d \in \{0,1\}$.  Plugging this expression into \Cref{eqn:imputation} implies that
\begin{align*}
    \Delta Y_{it^*}(0) &= Z_i'(\delta^*_{t^*} + \gamma_{t^*} \beta_{t^*}) + X_{it^*-1}(0)(\beta^*_{t^*} + \lambda_{t^*} \beta_{t^*}) + \beta_{t^*} u_{it^*} + \Delta v_{it^*} \\
    &:= Z_i'\delta^*_{2,t^*} + X_{it^*-1}(0) \beta^*_{2,t^*} + v_{2,it^*}
\end{align*}
where $\delta^*_{2,t^*} := \delta^*_{t^*} + \gamma_{t^*} \beta_{t^*}$, $\beta^*_{2,t^*} := \beta^*_{t^*} + \lambda_{t^*} \beta_{t^*}$ and $v_{2,it^*} := \beta_{t^*} u_{it^*} + \Delta v_{it^*}$.  Further, notice that $\E[v_{2,t^*} | Z, X_{t^*-1}, D=d] = 0$ for $d \in \{0,1\}$.  Thus, in this case, one can estimate $\delta^*_{2,t^*}$ and $\beta^*_{2,t^*}$ from a regression of the change in outcomes over time using the untreated group, and then estimate the ATT from the sample analogue of
\begin{align*}
    ATT = \E[\Delta Y_{t^*} | D=1] - \Big(\E[Z|D=1]'\delta^*_{2,t^*} + \E[X_{t^*-1}|D=1] \beta^*_{2,t^*}\Big)
\end{align*}
Thus, this particular case bypasses the need for actually estimating a separate model for the change in the time-varying covariate over time.  This is perhaps not surprising as these are the same conditions as in \Cref{sec:identification} where it was sufficient for the researcher to condition on the pre-treatment value of the covariates to recover the ATT.

\paragraph{Case 3: \Cref{ass:cov-unc} holds conditional on $\bm{X_{t^*-1},W_{t^*-1},Z}$} \ 

In this case,
\begin{align} \label{eqn:case3}
    \E[\Delta X_{t^*}(0) | D=1] &= \E\left[ \E[\Delta X_{t^*}(0) | Z, X_{t^*-1}, W_{t^*-1},D=1] | D=1\right] \\
    & = \E[Z|D=1]'\gamma_{t^*} + \E[X_{t^*-1}|D=1] \lambda_{t^*} + \E[W_{t^*-1}|D=1]'\xi_{t^*} \nonumber
\end{align}
where the first equality holds by the law of iterated expectations, and the second equality holds by \Cref{ass:cov-unc} and by assuming a linear model for the change in untreated covariates over time.  This suggests estimating $\E[\Delta X_{t^*}(0)|D=1]$ by running a regression of $\Delta X_{t^*}$ on $Z$, $X_{t^*-1}$, and $W_{t^*-1}$ using only untreated observations in order to estimate the parameters $\gamma_{t^*}$, $\lambda_{t^*}$, and $\xi_{t^*}$, and then to estimate $\E[\Delta X_{t^*}(0) | D=1]$ by using the sample analogue of the expression in \Cref{eqn:case3}.\footnote{In the special case (and perhaps leading case) considered in \Cref{cor:1} where $W_{t^*-1}$ includes the lagged outcome, $Y_{t^*-1}$ (in addition to all time-invariant covariates and the pre-treatment version of the covariates), one can follow this same procedure with $Y_{t^*-1}$ substituting for $W_{t^*-1}$.}



\paragraph{Case 4: Unconditional Parallel Trends holds for time-varying covariates} \ 

For this case, we assume that $\E[\Delta X_{t^*}(0) | D=1] = \E[\Delta X_{t^*}(0) | D=0]$.  It immediately follows that 
\begin{align*}
    ATT = \E[\Delta Y_t | D=1] - \Big( \E[Z|D=1]'\delta^*_{t^*} + \E[\Delta X_{t^*} | D=0] \beta_{t^*} + \E[X_{t^*-1}(0)|D=1]\beta^*_{t^*} \Big)
\end{align*}

This expression is very similar to the one in Case 1, except that one should use the change in untreated potential covariates \textit{for the untreated group}.

\paragraph{Case 5: Conditional Parallel Trends holds for time-varying covariates} \ 

For this case, we assume that $\E[\Delta X_{t^*}(0) | Z, D=1] = \E[\Delta X_{t^*}(0) | Z, D=0]$.  In this case,
\begin{align}
    \E[\Delta X_{t^*}(0) | D=1] &= \E\left[ \E[\Delta X_{t^*}(0) | Z, D=1] \big| D=1\right] \nonumber \\
    &= \E\left[ \E[\Delta X_{t^*} | Z, D=0] \big| D=1\right] \nonumber \\
    &= \E[Z|D=1]'\gamma_{t^*} \label{eqn:case5}
\end{align}
where the first equality holds by the law of iterated expectations, the second equality holds under conditional parallel trends for time-varying covariates, and the last equality holds under a linearity assumption.  This suggests estimating $\gamma_{t^*}$ by running a regression of $\Delta X_{t^*}$ on $Z$ using only untreated observations and then to estimate $\E[\Delta X_{t^*}(0)|D=1]$ from the sample analogue of \Cref{eqn:case5}.

\paragraph{Case 6: Conditional Parallel Trends Holds under Generic Parallel Trends Assumption} \ 

For this case, we assume that $\E[\Delta X_{t^*}(0) | Z, W_{t^*-1}, D=1] = \E[\Delta X_{t^*}(0) | Z, W_{t^*-1}, D=0]$.  In this case,
\begin{align}
    \E[\Delta X_{t^*}(0) | D=1] &= \E\left[ \E[\Delta X_{t^*}(0) | Z, W_{t^*-1}, D=1] | D=1\right] \nonumber \\
    &= \E[Z|D=1]'\gamma_{t^*} + \E[W_{t^*-1}|D=1]'\xi_{t^*} \label{eqn:case6}
\end{align}
where the first equality holds by the law of iterated expectations, and the second equality holds by the conditional parallel trends assumption used in this case and a linearity assumption.  Similarly to above, this suggests running a regression of $\Delta X_{t^*}$ on $Z$ and $W_{t^*-1}$ using only untreated observations to estimate $\gamma_{t^*}$ and $\xi_{t^*}$ and then to estimate $\E[\Delta X_{t^*}(0)|D=1]$ from the sample analogue of \Cref{eqn:case6}.

\bigskip

\bigskip

All of the approaches discussed in this section are substantially more robust than the TWFE regressions discussed in \Cref{sec:twfe}.  In particular, unlike TWFE regressions, they allow for the path of untreated potential outcomes to depend on (i) time-invariant covariates, (ii) the pre-treatment level of time-varying covariates, and (iii) the change in time-varying covariates over time that would have occurred if the treatment had not taken place.  Given any of a number of assumptions on the path of time-varying covariates in the absence of the treatment (as in Cases 1-6 above), they allow for the treatment to have an effect on time-varying covariates.  They allow for general forms of treatment effect heterogeneity; for example, they do not require conditional ATTs to be linear in covariates (as in \Cref{ass:bias-linearity}(a)) nor do they require any of the extra treatment effect homogeneity conditions for TWFE regressions in \Cref{prop:twfe-extra-assumptions}.  Finally, they do not require any linearity conditions for the propensity score as in \Cref{ass:bias-pscore}.  Relative to the approach discussed in \Cref{sec:identification}, the approaches considered in this section require linearity assumptions on the model for untreated potential outcomes and, in some cases, on a model for the change in untreated potential covariates over time.  The two main advantages of this approaches in this section are (i) parallel trends assumptions for time varying covariates can be strong enough to identify the ATT, and (ii) the approaches in this section are also easy to implement --- they only requiring running regressions and computing averages.  

\begin{remark}
  Even in the case where $\beta_t = \beta$ (i.e., that the effect of time-varying covariates does not change over time) the strategies proposed in this section would still result in improved estimators relative to the TWFE regressions considered in \Cref{sec:twfe} as they would still allow for parallel trends to depend on time invariant covariates, allow for general forms of treatment effect heterogeneity, and do not require assumptions on the propensity score.  In the special case where $\beta_t = \beta$ \textit{and} \Cref{ass:cov-exog} holds (so that covariates are not affected by the treatment), then the approaches proposed in this section coincide with the ``imputation estimators'' proposed in  \citet[see Assumption 1']{borusyak-jaravel-spiess-2021}, but, in general, our approach allows for the path of untreated potential outcomes to depend on both the level and change of time-varying covariates as well as for time-varying covariates to be affected by the treatment.
\end{remark}

\begin{remark} \label{rem:multi-periods}
  This section has continued to consider the case with exactly two periods, but it is straightforward to extend these arguments to multiple periods and variation in treatment timing by estimating models for untreated potential outcomes using all available untreated observations (these are observations both for units that do not participate in the treatment in any time period as well as pre-treatment time periods for units that become treated at some point).  Once the model for untreated potential outcomes has been estimated, one can ``impute'' untreated potential outcomes for treated observations, and weighted averages of differences between observed treated potential outcomes and imputed untreated potential outcomes correspond to various treatment effect parameters, depending on the weights chosen by the researcher.
\end{remark}

\begin{additional-material}
\section{Multiple Time Periods}

Many applications in economics have more than exactly two time periods available.  In this section, we extend previous results in two directions.  First, following \citet{callaway-santanna-2021}, we show how the previous arguments can be naturally extended to identifying group-time average treatment effects and then be aggregated into, for example, overall treatment effect parameters or event studies.

Second, I consider pre-testing the main identifying assumptions in the case where a researcher has access to more than one pre-treatment period.

* comment on common practice of just conditioning on covariates in the first period

* can you do joint pre-test??

* may be more efficient estimators using GMM, but we ignore this now

\begin{remark}
  Compare to \citet{imai-kim-wang-2018}, in our case weights can depend on $Y_{t-1}$.
\end{remark}

\section{Estimation}

\point{machine learning}

The main challenge with \Cref{eqn:att-identified} is a practical one --- $\E[\Delta Y_t|X_t,X_{t-1},D=0]$ is challenging to estimate.  For one reason, (although we abstract from this issue) in most realistic applications, the dimension of $X$ may be fairly large.  This means that nonparametric estimation may be challenging in practice.\footnote{This is a primary concern even in the case with time invariant covariates and is primary motivation for recent work on doubly robust estimators in the context of difference in differences (\citet{santanna-zhao-2020}).}  A second practical challenge is that, in many applications, $X_t$ and $X_{t-1}$ may be highly correlated with each other.  This suggests that approaches such as trying to estimate that conditional expectation using a linear model may perform poorly in practice.\footnote{That being said, one potentially promising approach here would be to use some kind of machine learning approach here.  \citet{chang-2020} studies difference in differences under a conditional parallel trends assumption using machine learning estimation strategies.  That paper does not explicitly consider time varying covariates, but the expression in \Cref{eqn:att-identified} fits into that framework simply by including both $X_t$ and $X_{t-1}$ as covariates.}  Indeed, trying to estimate this sort of expression either parametrically, nonparametrically, or using machine learning is not common in applied work.

\section{Application}

Good candidates for applications:

\begin{itemize}
    \item job displacement
    \item covid - \point{it would be really helpful for our arguments if we could find a paper that was interested in the effect of some Covid-related policy on some outcome (other than Covid cases) but wanted to control for the number of Covid cases in their regressions.}
    \item cheng and hoekstra's castle doctrine paper; they use state level data, and include a number of time varying covariates.  A quick looks suggests that these would be unlikely to be affected by the treatment itself.  However, if they only have state-level data, it may be practically challenging to implement our approach on their data.
\end{itemize}
\end{additional-material}


\section{Conclusion}

In the current paper, we have considered DID identification strategies where the parallel trends assumption holds only after conditioning on time varying covariates that may themselves be affected by the treatment.  This setting is common in empirical applications in economics, and we have provided several approaches that offer a number of advantages relative to more commonly used TWFE regressions that include covariates (even in the case where there are only two time periods).  In addition, the new approaches that we have proposed are generally not much more complicated to implement than TWFE regressions.

\onehalfspace

\printbibliography

@unpublished{gupta-montenovo-nguyen-rojas-schmutte-simon-weinburg-2020,
  title={Effects of social distancing policy on labor market outcomes},
  author={Gupta, Sumedha and Montenovo, Laura and Nguyen, Thuy Dieu and Lozano-Rojas, Felipe and Schmutte, Ian M and Simon, Kosali Ilayperuma and Weinberg, Bruce A and Wing, Coady},
  note={Working Paper},
  year={2020}
}

@article{chetty-friedman-hendren-stepner-2020,
  title={The economic impacts of COVID-19: Evidence from a new public database built from private sector data},
  author={Chetty, Raj and Friedman, J and Hendren, Nathaniel and Stepner, Michael},
  journal={Opportunity Insights},
  year={2020}
}

@unpublished{weill-stigler-deschenes-springborn-2021,
  title={Researchers' Degrees-of-Flexibility and the Credibility of Difference-in-Differences Estimates: Evidence From the Pandemic Policy Evaluations},
  author={Weill, Joakim A and Stigler, Matthieu and Deschenes, Olivier and Springborn, Michael R},
  year={2021},
  note={Working Paper}
}

@article{jacobson-lalonde-sullivan-1993,
  title={Earnings losses of displaced workers},
  author={Jacobson, Louis S and LaLonde, Robert J and Sullivan, Daniel G},
  journal={American Economic Review},
  pages={685--709},
  year={1993},
  publisher={JSTOR}
}

@article{stevens-1997,
  title={Persistent effects of job displacement: The importance of multiple job losses},
  author={Stevens, Ann Huff},
  journal={Journal of Labor Economics},
  pages={165--188},
  year={1997},
  publisher={JSTOR}
}

@article{topel-1991,
  title={Specific capital, mobility, and wages: Wages rise with job seniority},
  author={Topel, Robert},
  journal={Journal of Political Economy},
  volume={99},
  number={1},
  pages={145--176},
  year={1991},
  publisher={The University of Chicago Press}
}

@article{brand-2006,
  title={The effects of job displacement on job quality: Findings from the Wisconsin Longitudinal Study},
  author={Brand, Jennie E},
  journal={Research in Social Stratification and Mobility},
  volume={24},
  number={3},
  pages={275--298},
  year={2006},
  publisher={Elsevier}
}

@article{barnette-odongo-reynolds-2021,
  title={Changes over time in the cost of job loss for young men and women},
  author={Barnette, Justin and Odongo, Kennedy and Reynolds, C Lockwood},
  journal={The BE Journal of Economic Analysis \& Policy},
  volume={21},
  number={1},
  pages={335--378},
  year={2021},
  publisher={De Gruyter}
}

@article{cheng-hoekstra-2013,
  title={Does strengthening self-defense law deter crime or escalate violence? Evidence from expansions to castle doctrine},
  author={Cheng, Cheng and Hoekstra, Mark},
  journal={Journal of Human Resources},
  volume={48},
  number={3},
  pages={821--854},
  year={2013},
  publisher={University of Wisconsin Press}
}

@article{bonhomme-sauder-2011,
  title={Recovering distributions in difference-in-differences models: {A} comparison of selective and comprehensive schooling},
  author={Bonhomme, Stephane and Sauder, Ulrich},
  journal={Review of Economics and Statistics},
  volume={93},
  number={2},
  pages={479--494},
  year={2011},
  publisher={MIT Press}
}

@article{athey-imbens-2006,
  title={Identification and inference in nonlinear difference-in-differences models},
  author={Athey, Susan and Imbens, Guido},
  journal={Econometrica},
  volume={74},
  number={2},
  pages={431--497},
  year={2006},
  publisher={Wiley Online Library}
}

@article{abadie-2005,
  title={Semiparametric difference-in-differences estimators},
  author={Abadie, Alberto},
  journal={The Review of Economic Studies},
  volume={72},
  number={1},
  pages={1--19},
  year={2005},
  publisher={Oxford University Press}
}

@article{heckman-ichimura-todd-1997,
  title={Matching as an econometric evaluation estimator: {Evidence} from evaluating a job training programme},
  author={Heckman, James and Ichimura, Hidehiko and Todd, Petra},
  journal={The Review of Economic Studies},
  volume={64},
  number={4},
  pages={605--654},
  year={1997},
  publisher={Oxford University Press}
}

@article{heckman-ichimura-smith-todd-1998,
  title={Characterizing selection bias using experimental data},
  author={Heckman, James and Ichimura, Hidehiko and Smith, Jeffrey and Todd, Petra},
  journal={Econometrica},
  volume={66},
  number={5},
  pages={1017--1098},
  year={1998}
}

@article{imbens-wooldridge-2009,
  title={Recent developments in the econometrics of program evaluation},
  author={Imbens, Guido and Wooldridge, Jeffrey},
  journal={Journal of Economic Literature},
  volume={47},
  number={1},
  pages={5--86},
  year={2009}
}

@article{callaway-li-2019,
  title={Quantile treatment effects in difference in differences models with panel data},
  author={Callaway, Brantly and Li, Tong},
  journal={Quantitative Economics},
  volume={10},
  number={4},
  pages={1579--1618},
  year={2019},
  publisher={Wiley Online Library}
}

@unpublished{melly-santangelo-2015,
  title={The changes-in-changes model with covariates},
  author={Melly, Blaise and Santangelo, Giulia},
  year={2015},
  note={Working Paper}
}

@article{callaway-li-oka-2018,
  title={Quantile treatment effects in difference in differences models under dependence restrictions and with only two time periods},
  author={Callaway, Brantly and Li, Tong and Oka, Tatsushi},
  journal={Journal of Econometrics},
  year={2018},
  volume={206},
  number={2},
  pages={395--413},
  publisher={Elsevier}
}

@article{bierens-1982,
  title={Consistent model specification tests},
  author={Bierens, Herman J},
  journal={Journal of Econometrics},
  volume={20},
  number={1},
  pages={105--134},
  year={1982},
  publisher={Elsevier}
}

@article{stute-1997,
  title={Nonparametric model checks for regression},
  author={Stute, Winfried},
  journal={The Annals of Statistics},
  pages={613--641},
  year={1997},
  publisher={JSTOR}
}

@book{angrist-pischke-2008,
  title={Mostly Harmless Econometrics: An Empiricist's Companion},
  author={Angrist, Joshua D and Pischke, Jorn-Steffen},
  year={2008},
  publisher={Princeton University Press}
}

@article{lechner-2011,
  title={The estimation of causal effects by difference-in-difference methods},
  author={Lechner, Michael},
  journal={Foundations and Trends in Econometrics},
  volume={4},
  number={3},
  pages={165--224},
  year={2011},
  publisher={Now Publishers, Inc.}
}

@article{blundell-dias-2009,
  title={Alternative approaches to evaluation in empirical microeconomics},
  author={Blundell, Richard and Costa Dias, Monica},
  journal={Journal of Human Resources},
  volume=44,
  number=3,
  pages={565--640},
  year=2009,
  publisher={University of Wisconsin Press}
}

@article{callaway-santanna-2021,
  title={Difference-in-differences with multiple time periods},
  author={Callaway, Brantly and Sant'Anna, Pedro HC},
  year={2021},
  journal={Journal of Econometrics},
  volume={225},
  number={2},
  pages={200--230}
}

@article{goodman-2021,
  title={Difference-in-differences with variation in treatment timing},
  author={Goodman-Bacon, Andrew},
  year={2021},
  journal={Journal of Econometrics},
  volume={225},
  number={2},
  pages={254--277}
}

@article{chaisemartin-dhaultfoeuille-2020,
  title={Two-way fixed effects estimators with heterogeneous treatment effects},
  author={{de Chaisemartin}, Clement and D'Haultf{\oe}uille, Xavier},
  year={2020},
  journal={American Economic Review},
  volume={110},
  number={9},
  pages={2964--2996}
}

@unpublished{borusyak-jaravel-spiess-2021,
  title={Revisiting event study designs: Robust and efficient estimation},
  author={Borusyak, Kirill and Jaravel, Xavier and Spiess, Jann},
  note={Working Paper},
  year={2021}
}

@unpublished{imai-kim-wang-2018,
  title={Matching Methods for Causal Inference with Time-Series Cross-Section Data},
  author={Imai, Kosuke and Kim, In Song and Wang, Erik},
  year={2018},
  note={Working Paper}
}

@article{robins-hernan-brumback-2000,
  title={Marginal structural models and causal inference in Epidemiology},
  author={Robins, James M and Hern{\'a}n, Miguel Angel and Brumback, Babette},
  journal={Epidemiology},
  volume={11},
  number={5},
  pages={551},
  year={2000}
}

@article{blackwell-glynn-2018,
  title={How to make causal inferences with time-series cross-sectional data under selection on observables},
  author={Blackwell, Matthew and Glynn, Adam N},
  journal={American Political Science Review},
  volume={112},
  number={4},
  pages={1067--1082},
  year={2018},
  publisher={Cambridge University Press}
}

@incollection{robins-1997,
  title={Causal inference from complex longitudinal data},
  author={Robins, James M},
  booktitle={Latent variable modeling and applications to causality},
  pages={69--117},
  year={1997},
  publisher={Springer}
}

@article{lechner-2008,
  title={A note on endogenous control variables in causal studies},
  author={Lechner, Michael},
  journal={Statistics \& Probability Letters},
  volume={78},
  number={2},
  pages={190--195},
  year={2008},
  publisher={Elsevier}
}

@article{santanna-zhao-2020,
  title={Doubly robust difference-in-differences estimators},
  author={Sant’Anna, Pedro HC and Zhao, Jun},
  journal={Journal of Econometrics},
  volume={219},
  number={1},
  pages={101--122},
  year={2020},
  publisher={Elsevier}
}

@article{robins-rotnitzky-zhao-1994,
  title={Estimation of regression coefficients when some regressors are not always observed},
  author={Robins, James M and Rotnitzky, Andrea and Zhao, Lue Ping},
  journal={Journal of the American statistical Association},
  volume={89},
  number={427},
  pages={846--866},
  year={1994},
  publisher={Taylor \& Francis}
}

@article{scharfstein-rotnitzky-robins-1999,
  title={Adjusting for nonignorable drop-out using semiparametric nonresponse models},
  author={Scharfstein, Daniel O and Rotnitzky, Andrea and Robins, James M},
  journal={Journal of the American Statistical Association},
  volume={94},
  number={448},
  pages={1096--1120},
  year={1999},
  publisher={Taylor \& Francis Group}
}

@article{sloczynski-wooldridge-2018,
  title={A general double robustness result for estimating average treatment effects},
  author={S{\l}oczy{\'n}ski, Tymon and Wooldridge, Jeffrey M},
  journal={Econometric Theory},
  volume={34},
  number={1},
  pages={112--133},
  year={2018},
  publisher={Cambridge University Press}
}

@incollection{arellano-honore-2001,
  title={Panel data models: some recent developments},
  author={Arellano, Manuel and Honor{\'e}, Bo},
  booktitle={Handbook of Econometrics},
  volume={5},
  pages={3229--3296},
  year={2001},
  publisher={Elsevier}
}

@unpublished{callaway-li-2021b,
  title={Policy evaluation during a pandemic},
  author={Callaway, Brantly and Li, Tong},
  year={2021},
  note={Working Paper}
}

@unpublished{liu-wang-xu-2021,
  title={A practical guide to counterfactual estimators for causal inference with time-series cross-sectional data},
  author={Liu, Licheng and Wang, Ye and Xu, Yiqing},
  note={Working Paper},
  year={2021}
}

@unpublished{gardner-2021,
  title={Two-stage difference in differences},
  author={Gardner, John},
  note={Working Paper},
  year={2021}
}

@unpublished{chabe-2017,
  title={Should we combine difference in differences with conditioning on pre-treatment outcomes?},
  author={Chab{\'e}-Ferret, Sylvain},
  year={2017},
  note={Working paper}
}

@article{chernozhukov-etal-2018,
  title={Double/debiased machine learning for treatment and structural parameters},
  author={Chernozhukov, Victor and Chetverikov, Denis and Demirer, Mert and Duflo, Esther and Hansen, Christian and Newey, Whitney and Robins, James},
  journal={The Econometrics Journal},
  volume={21},
  number={1},
  pages={C1--C68},
  year={2018},
  publisher={Wiley Online Library}
}

@article{chang-2020,
  title={Double/debiased machine learning for difference-in-differences models},
  author={Chang, Neng-Chieh},
  journal={The Econometrics Journal},
  volume={23},
  number={2},
  pages={177--191},
  year={2020},
  publisher={Oxford University Press}
}

@article{zeldow-hatfield-2021,
  title={Confounding and regression adjustment in difference-in-differences studies},
  author={Zeldow, Bret and Hatfield, Laura A},
  journal={Health Services Research},
  year={2021},
  publisher={Wiley Online Library}
}

@article{rosenbaum-1984,
  title={The consequences of adjustment for a concomitant variable that has been affected by the treatment},
  author={Rosenbaum, Paul R},
  journal={Journal of the Royal Statistical Society: Series A (General)},
  volume={147},
  number={5},
  pages={656--666},
  year={1984},
  publisher={Wiley Online Library}
}

@article{flores-lagunes-2009,
  title={Identification and estimation of causal mechanisms and net effects of a treatment under unconfoundedness},
  author={Flores, Carlos A and Flores-Lagunes, Alfonso},
  year={2009},
  publisher={IZA Discussion paper}
}

@inproceedings{huber-2020,
  author={Huber, Martin},
  editor={Zimmermann, Klaus F.},
  title={Mediation Analysis},
  booktitle={Handbook of Labor, Human Resources and Population Economics},
  pages={1--38},
  year={2020},
  publisher={Springer International Publishing}
}

@article{angrist-1998,
  title={Estimating the Labor Market Impact of Voluntary Military Service Using Social Security Data on Military Applicants},
  author={Angrist, Joshua D},
  journal={Econometrica},
  volume={66},
  number={2},
  pages={249--288},
  year={1998}
}

@article{aronow-samii-2016,
  title={Does regression produce representative estimates of causal effects?},
  author={Aronow, Peter M and Samii, Cyrus},
  journal={American Journal of Political Science},
  volume={60},
  number={1},
  pages={250--267},
  year={2016},
  publisher={Wiley Online Library}
}

@unpublished{goldsmith-hull-kolesar-2021,
  title={On estimating multiple treatment effects with regression},
  author={Goldsmith-Pinkham, Paul and Hull, Peter and Koles{\'a}r, Michal},
  year={2021},
  note={Working Paper}
}

@unpublished{ishimaru-2021,
  title={Empirical Decomposition of the IV-OLS Gap with Heterogeneous and Nonlinear Effects},
  author={Ishimaru, Shoya},
  year={2021},
  note={Working Paper}
}

@unpublished{ishimaru-2022,
  title={What Do We Get From A Two-Way Fixed Effects Estimator? Implications From A General Numerical Equivalence},
  author={Ishimaru, Shoya},
  year={2022},
  note={Working Paper}
}

@article{sloczynski-2020,
  title={Interpreting OLS estimands when treatment effects are heterogeneous: Smaller groups get larger weights},
  author={S{\l}oczy{\'n}ski, Tymon},
  journal={The Review of Economics and Statistics},
  pages={1--27},
  year={2020}
}

\appendix\renewcommand{\theequation}{A\arabic{equation}}\setcounter{equation}{0}

\section{Proofs} \label{app:proofs}

\subsection*{Proof of \Cref{thm:1}}

\begin{proof}
  To start with, notice that 
  \begin{align*}
      ATT &= \E[Y_{t^*}(1) - Y_{t^*}(0) | D=1] \\
      &= \E[Y_{t^*}(1) - Y_{t^*-1}(0) | D=1] + \E[Y_{t^*}(0) - Y_{t^*-1}(0) | D=1] \\
      &= \E[\Delta Y_{t^*} | D=1] - \E[\Delta Y_{t^*}(0) | D=1]
  \end{align*}
  where the first equality is just the definition of $ATT$, the second equality holds by adding and subtracting $\E[Y_{t^*-1}(0) | D=1]$, and the third equality holds by writing potential outcomes in terms of their observed counterparts.  For part (1), further notice that,
  \begin{align*}
      \E[\Delta Y_{t^*}(0) | D=1] &= \E\Big[ \E[\Delta Y_{t^*}(0) | X_{t^*}(0), X_{t^*-1}, Z, D=1] \Big| D=1 \Big] \\
      &= \E\Big[ \E[\Delta Y_{t^*}(0) | X_{t^*}(0), X_{t^*-1}, Z, D=0] \Big| D=1 \Big] \\
      &= \E\Big[ \E[\Delta Y_{t^*} | X_{t^*}, X_{t^*-1}, Z, D=0] \Big| D=1 \Big]
  \end{align*}
  where the first equality holds by the law of iterated expectations, the second equality holds by \Cref{ass:conditional-parallel-trends}, and the last equality holds because $\Delta Y_{t^*}(0)$ and $X_{t^*}(0)$ are observed for the untreated group and uses \Cref{ass:cov-exog} to integrate over the distribution of observed covariates (i.e., treated potential covariates) for the treated group.  Combining this expression with the previous one for $ATT$ completes the proof for part (1) of the result.
  
  For part (2), notice that
  \begin{align*}
      \E[\Delta Y_{t^*}(0) | D=1] &= \E\Big[ \E[\Delta Y_{t^*}(0) | X_{t^*}(0), X_{t^*-1}, Z, D=1] \Big| D=1 \Big] \\
      &= \E\Big[ \E[\Delta Y_{t^*}(0) | X_{t^*}(0), X_{t^*-1}, Z, D=0] \Big| D=1 \Big] \\
      &= \E\Big[ \E\big[\E[\Delta Y_{t^*}(0) | X_{t^*}(0), X_{t^*-1}, Z, D=0] \big| X_{t^*-1},W_{t^*-1},Z, D=1 \big] \Big| D=1 \Big] \\
      &= \E\Big[ \E\big[\E[\Delta Y_{t^*}(0) | X_{t^*}(0), X_{t^*-1}, Z, D=0] \big| X_{t^*-1},W_{t^*-1},Z, D=0 \big] \Big| D=1 \Big] \\
      &= \E\Big[ \E\big[\E[\Delta Y_{t^*} | X_{t^*}, X_{t^*-1}, Z, D=0] \big| X_{t^*-1},W_{t^*-1}, Z, D=0 \big] \Big| D=1 \Big]
  \end{align*}
  where the first equality holds by the law of iterated expectations, the second equality holds by \Cref{ass:conditional-parallel-trends} (unlike part (1), this term is not immediately identified because we do not have an immediate analogue of the distribution of $X_{t^*}(0)$ in order to identify the outer expectation), the third equality holds by the law of iterated expectations, the fourth equality holds by \Cref{ass:cov-unc} (because, after conditioning on $(X_{t^*-1},W_{t^*-1},Z)$, the only randomness comes from $X_{t^*}(0)$), the fifth equality holds by writing potential outcomes in terms of their observed counterparts, and this term is identified because the distribution of $(X_{t^*-1},W_{t^*-1},Z)$ is identified for the treated group.
\end{proof}

\subsection*{Proof of \Cref{cor:1}}

\begin{proof}
For part (1), the result holds immediately by the law of iterated expectations.  For part (2), the result holds immediately from the expression in part (2) of \Cref{thm:1} using $W_{t^*-1}=Y_{t^*-1}$.
\end{proof}

\subsection*{Proof of \Cref{cor:2}}

\begin{proof}
  For part (1), we omit the proof as, after invoking \Cref{ass:cov-exog}, this becomes the same case as with time invariant covariates --- see, for example, \citet{santanna-zhao-2020} for this sort of result in the case with time invariant covariates.  Given the expression for the ATT in part (1) of \Cref{cor:1}, the proof of part (2) follows using the same arguments as for part (1).
\end{proof}

\subsection*{Proof of \Cref{prop:twfe}}
We prove the result in several steps.

To start with, consider the numerator in the expression for $\alpha$ in \Cref{eqn:twfe}.  Notice that
\begin{align}
    \E[De] &= \E[D(\Delta Y_{t^*} - \L(\Delta Y_{t^*} | \Delta X_{t^*})] \nonumber \\
    &= \E[D(\Delta Y_{t^*} - \E[\Delta Y_{t^*} | X^{all}(0)] )] + \E[D(\E[\Delta Y_{t^*}|X^{all}(0)] - \L(\Delta Y_{t^*} | \Delta X_{t^*})] \label{eqn:twfe2}
\end{align}
We provide results for each of the terms in \Cref{eqn:twfe2} next.
\begin{lemma} \label{lem:twfe-num1} Under Assumptions \ref{ass:sampling}, \ref{ass:conditional-parallel-trends}, and \ref{ass:overlap}(a), 
  \begin{align*}
      \E[D(\Delta Y_{t^*} - \E[\Delta Y_{t^*} | X^{all}(0)])] =  \E\left[ \E[D] (1- p(X^{all}(0))) ATT_{X^{all}(0)}(X^{all}(0)) \Big| D=1 \right] 
  \end{align*}
\end{lemma}
\begin{proof}
  \begin{align*}
     & \E[D(\Delta Y_{t^*} - \E[\Delta Y_{t^*} | X^{all}(0)])] \\
     & \hspace{10pt} = \E\left[D\left(\Delta Y_{t^*} - (\E[\Delta Y_{t^*} | X^{all}(0), D=1] p(X^{all}(0)) - \E[\Delta Y_{t^*} | X^{all}(0), D=0] (1-p(X^{all}(0))) \right) \right] \\
     & \hspace{10pt} = \E\left[ \E[D \Delta Y_{t^*} | X^{all}(0)] \right. \\
     & \hspace{50pt} \left. - p(X^{all}(0)) \left(\E[\Delta Y_{t^*} | X^{all}(0), D=1] p(X^{all}(0)) - \E[\Delta Y_{t^*} | X^{all}(0), D=0] (1-p(X^{all}(0))\right) \right] \\
     & \hspace{10pt} = \E\left[ \E[\Delta Y_{t^*} | X^{all}(0), D=1] p(X^{all}(0)) \right. \\
     & \hspace{50pt} \left. - p(X^{all}(0)) \left(\E[\Delta Y_{t^*} | X^{all}(0), D=1] p(X^{all}(0)) - \E[\Delta Y_{t^*} | X^{all}(0), D=0] (1-p(X^{all}(0))\right) \right] \\
     & \hspace{10pt} = \E\left[ p(X^{all}(0)) (1- p(X^{all}(0))) \left( \E[\Delta Y_{t^*} | X^{all}(0), D=1] - \E[\Delta Y_{t^*} | X^{all}(0), D=0] \right) \right] \\
     & \hspace{10pt} = \E\left[ \E[D] (1- p(X^{all}(0))) \left( \E[\Delta Y_{t^*} | X^{all}(0), D=1] - \E[\Delta Y_{t^*} | X^{all}(0), D=0] \right) \Big| D=1 \right] \\
     & \hspace{10pt} = \E\left[ \E[D] (1- p(X^{all}(0))) ATT_{X^{all}(0)}(X^{all}(0)) \Big| D=1 \right] 
  \end{align*}
  where the first three equalities hold by repeatedly applying the law of iterated expectations, the fourth equality holds by rearranging terms, the fifth equality holds by integrating over the distribution of $X^{all}(0)$ conditional on $D=1$ and re-weighting, and the last equality holds under the conditional parallel trends assumption in \Cref{ass:conditional-parallel-trends}.
\end{proof}

Next, we provide a result for the second term in \Cref{eqn:twfe2}.  

\begin{lemma} \label{lem:twfe-num2} Under Assumptions \ref{ass:sampling}, \ref{ass:conditional-parallel-trends}, and \ref{ass:overlap}(a),
\begin{align*}
  & \E[D(\E[\Delta Y_{t^*}|X^{all}(0)] - \L(\Delta Y_{t^*} | \Delta X_{t^*})] \\
  & \hspace{25pt} = \E\left[ \E[D] \left\{\Big( \E[\Delta Y_{t^*} | X^{all}(0), D=1] - \L(\Delta Y_{t^*} | \Delta X_{t^*}, D=1) \Big) p(X^{all}(0)) \right. \right.\\
    & \hspace{75pt} \left.\left.-  \Big(\E[\Delta Y_{t^*} | X^{all}(0), D=0] - \L(\Delta Y_{t^*} | \Delta X_{t^*}, D=0) \Big) (1-p(X^{all}(0))) \right\} \Big| D=1 \right] \\
    & \hspace{25pt} + \E\left[ \E[D] \left\{ \L(\Delta Y_{t^*} | \Delta X_{t^*}, D=1) - \L(\Delta Y_{t^*} | \Delta X_{t^*}, D=0)\right\}(p(X^{all}(0)) - \L(D|\Delta X_{t^*})) \Big| D=1 \right]
\end{align*}
\end{lemma}
\begin{proof}
Notice that
 \begin{align*}
  & \E[D(\E[\Delta Y_{t^*}|X^{all}(0)] - \L(\Delta Y_{t^*} | \Delta X_{t^*})] \\
    & \hspace{25pt} = \E[ p(X^{all}(0)) (\E[\Delta Y_{t^*}|X^{all}(0)] - \L(\Delta Y_{t^*} | \Delta X_{t^*}) ] \\
    & \hspace{25pt} = \E[ p(X^{all}(0)) \{ \E[\Delta Y_{t^*} | X^{all}(0), D=1] p(X^{all}(0)) + \E[\Delta Y_{t^*} |  X^{all}(0), D=0](1-p(X^{all}(0))] \}] \\
    & \hspace{50pt} - \E[ p(X^{all}(0)) \{ \L(\Delta Y_{t^*} | \Delta X_{t^*},D=1)\L(D|\Delta X_{t^*}) + \L(\Delta Y_{t^*} | \Delta X_{t^*},D=0) (1-\L(D|\Delta X_{t^*}) \}] \\
    & \hspace{25pt} = \E\left[p(X^{all}(0)) \left\{\Big( \E[\Delta Y_{t^*} | X^{all}(0), D=1] - \L(\Delta Y_{t^*} | \Delta X_{t^*}, D=1) \Big) p(X^{all}(0)) \right. \right.\\
    & \hspace{100pt} \left.\left.+  \Big(\E[\Delta Y_{t^*} | X^{all}(0), D=0] - \L(\Delta Y_{t^*} | \Delta X_{t^*}, D=0) \Big) (1-p(X^{all}(0))) \right\} \right] \\
    & \hspace{25pt} + \E\left[p(X^{all}(0))\left\{ \L(\Delta Y_{t^*} | \Delta X_{t^*}, D=1) - \L(\Delta Y_{t^*} | \Delta X_{t^*}, D=0)\right\}(p(X^{all}(0)) - \L(D|\Delta X_{t^*}))\right] \\
    & \hspace{25pt} = \E\left[ \E[D] \left\{\Big( \E[\Delta Y_{t^*} | X^{all}(0), D=1] - \L(\Delta Y_{t^*} | \Delta X_{t^*}, D=1) \Big) p(X^{all}(0)) \right. \right.\\
    & \hspace{100pt} \left.\left.+  \Big(\E[\Delta Y_{t^*} | X^{all}(0), D=0] - \L(\Delta Y_{t^*} | \Delta X_{t^*}, D=0) \Big) (1-p(X^{all}(0))) \right\} \Big| D=1 \right] \\
    & \hspace{25pt} + \E\left[ \E[D] \left\{ \L(\Delta Y_{t^*} | \Delta X_{t^*}, D=1) - \L(\Delta Y_{t^*} | \Delta X_{t^*}, D=0)\right\}(p(X^{all}(0)) - \L(D|\Delta X_{t^*})) \Big| D=1 \right] 
\end{align*}
where the first equality holds by applying the law of iterated expectations, the second equality holds by applying the law of iterated expectations and the law of iterated projections, the third equality holds by adding and subtracting $\E[\L(\Delta Y_{t^*} | \Delta X_{t^*}, D=1)p(X^{all}(0))^2]$ and $\E[\L(\Delta Y_{t^*} | \Delta X_{t^*}, D=0)p(X^{all}(0))(1-p(X^{all}(0)))]$ and rearranging terms, the fourth equality holds by applying the law of iterated expectations to each each term.  This completes the proof.
\end{proof}

Next, we provide a result on decomposing differences between the conditional expectation of $\Delta Y_{t^*}$ (conditional on the full vector $X^{all}(0)$) and the linear projection of $\Delta Y_{t^*}$ on $\Delta X_{t^*}$.

\begin{lemma} \label{lem:twfe-el-diff} Under Assumptions \ref{ass:sampling}, \ref{ass:conditional-parallel-trends}, and \ref{ass:overlap}(a) and for $d \in \{0,1\}$,  
\begin{align}
    & \E[\Delta Y_{t^*} | X^{all}(0), D=d] - \L(\Delta Y_{t^*} | \Delta X_{t^*}, D=d) \nonumber \\
    & \hspace{50pt} = \E[\Delta Y_{t^*} | X_t(0), X_{t-1}, Z, D=d] - \E[\Delta Y_{t^*} | X_t, X_{t-1}, Z, D=d] \tag{A}\\
    & \hspace{75pt} + \E[\Delta Y_{t^*} | X_t, X_{t-1}, Z, D=d] - \E[\Delta Y_{t^*} | X_t, X_{t-1}, D=d] \tag{B} \\
    & \hspace{75pt} + \E[\Delta Y_{t^*} | X_t, X_{t-1}, D=d] - \E[\Delta Y_{t^*} | \Delta X_{t^*}, D=d] \tag{C} \\
    & \hspace{75pt} + \E[\Delta Y_{t^*} | \Delta X_{t^*}, D=d] - \L(\Delta Y_{t^*} | \Delta X_{t^*}, D=d) \tag{D}
\end{align}
\end{lemma}

\point{it is not 100\% clear to me what the best order to add and subtract terms is.  Particularly, we could go from $X_{t^*}(0)$ to $X_{t^*}$ towards the end rather than the beginning.  This might somewhat alter the assumptions that we need to make in a favorable way (e.g., we might not require \Cref{ass:cov-exog} for \textit{every} condition to hold.  The drawback would be that more of these expressions would contain terms that are not necessarily identified.}

\begin{proof}
  The result holds immediately just by adding and subtracting terms.
\end{proof}

Next, we provide a useful result for the denominator in the expression for $\alpha$ in \Cref{eqn:twfe}.

\begin{lemma} \label{lem:twfe-denom} Under Assumptions \ref{ass:sampling}, \ref{ass:conditional-parallel-trends}, and \ref{ass:overlap}(a), 
  \begin{align*}
    \E[u^2] &= \E[ p(\Delta X_{t^*}) (1-\L(D|\Delta X_{t^*})) ] \\
    &= \E\big[ \E[D] (1-\L(D|\Delta X_{t^*})) | D=1 \big]
  \end{align*}
\end{lemma}

\begin{proof}
  From the definition of $u$, it follows that
  \begin{align}
    \E\left[u^2\right] &= \E\left[ (D - \L(D|\Delta X_{t^*}))^2\right] \nonumber \\
                       &= \E[D] - 2 \E[D \L(D|\Delta X_{t^*})] + \E[\L(D|\Delta X_{t^*})^2] \nonumber \\
                       &:= A_1 - 2A_2 + A_3 \label{eqn:a2a}
  \end{align}
  and we consider each of these in turn. Start with,
  \begin{align}
    A_2 &= \E[ D \L(D|\Delta X_{t^*}) ] \nonumber \\
        &= \E[ D \Delta X_{t^*}' \E[\Delta X_{t^*}\Delta X_{t^*}']^{-1} \E[\Delta X_{t^*}D] \nonumber \\
        &= \E[\Delta X_{t^*} D]'\E[\Delta X_{t^*}\Delta X_{t^*}']^{-1} \E[\Delta X_{t^*}D] \label{eqn:a2}
  \end{align}
  where the first equality holds from the definition of $\L(D|\Delta X_{t^*})$ and the second equality holds immediately from the previous one.  Next,
  \begin{align}
    A_3 &= \E[\L(D|\Delta X_{t^*})^2] \nonumber \\
        &= \E\left[ \E[\Delta X_{t^*} D]' \E[\Delta X_{t^*}\Delta X_{t^*}']^{-1} \Delta X_{t^*} \Delta X_{t^*}' \E[\Delta X_{t^*}\Delta X_{t^*}']^{-1} \E[\Delta X_{t^*} D] \right] \nonumber\\
    &= \E[\Delta X_{t^*} D]'\E[\Delta X_{t^*}\Delta X_{t^*}']^{-1} \E[\Delta X_{t^*}D] \label{eqn:a3}
  \end{align}
  where the first equality holds by the definition of $A_3$, the second equality holds by the definition of $\L(D|\Delta X_{t^*})$, and the last equality holds by canceling terms.  Plugging \Cref{eqn:a2,eqn:a3} in \Cref{eqn:a2a} implies that
  \begin{align*}
    \E[u^2] &= \E[ D(1-\L(D|\Delta X_{t^*})) ] \\
    &= \E[ p(\Delta X_{t^*}) (1-\L(D|\Delta X_{t^*})) ] \\
    &= \E[\E[D](1-\L(D|\Delta X_{t^*})) | D=1]
  \end{align*}
  where the second and third equalities hold by the law of iterated expectations and which completes the proof.
\end{proof}

\begin{proof}[Proof of \Cref{prop:twfe}]
    The first part of the expression for $\alpha$ comes from \Cref{eqn:twfe2} and by \Cref{lem:twfe-num1} and \Cref{lem:twfe-denom}.  The second and third parts come from \Cref{eqn:twfe2} and by \Cref{lem:twfe-num2,lem:twfe-el-diff,lem:twfe-denom}.
\end{proof}

\subsection*{Proof of \Cref{prop:twfe-extra-assumptions}}

To show the first part of the result, we show that each of Terms (A)-(D) in \Cref{prop:twfe} are equal to zero under the extra conditions in this proposition.

\paragraph{Term (A):} First, consider the expression in Term (A) for $d=0$.  Notice that,
\begin{align*}
    \E[\Delta Y_{t^*} | X_{t^*}(0), X_{t^*-1}, Z, D=0] &= \E[\Delta Y_{t^*} | X_{t^*}, X_{t^*-1}, Z, D=0]
\end{align*}
which holds because untreated potential covariates are equal to observed covariates for the untreated group.  Second, consider the case when $d=1$.  In this case,
\begin{align*}
    \E[\Delta Y_{t^*} | X_{t^*}(0), X_{t^*-1}, Z, D=1] &= \E[\Delta Y_{t^*} | X_{t^*}, X_{t^*-1}, Z, D=1]
\end{align*}
holds immediately by \Cref{ass:cov-exog}.  Thus, Term (A) is equal to zero under \Cref{ass:cov-exog}.

\paragraph{Term (B):} First, consider the expression in Term (B) for $d=0$.  Notice that,
\begin{align*}
    \E[\Delta Y_{t^*} | X_{t^*}, X_{t^*-1}, Z, D=0] &= \E[\Delta Y_{t^*} | X_{t^*}, X_{t^*-1}, D=0]
\end{align*}
under \Cref{ass:bias-z}(b) because, conditional on being in the treated group, the observed change in outcomes is equal to the change in untreated potential outcomes and the observed $X_{t^*}$ is equal to $X_{t^*}(0)$.  Second, consider the case when $d=1$, in this case
\begin{align*}
    \E[\Delta Y_{t^*} | X_{t^*}, X_{t^*-1}, Z, D=1] &= \Big(\E[\Delta Y_{t^*} | X_{t^*}, X_{t^*-1}, Z, D=1] - \E[\Delta Y_{t^*}(0) | X_{t^*}(0), X_{t^*-1}, Z, D=1]\Big) \\
    &\hspace{25pt} + \Big(\E[\Delta Y_{t^*}(0) | X_{t^*}(0), X_{t^*-1}, Z, D=1] - \E[\Delta Y_{t^*}(0) | X_{t^*}(0), X_{t^*-1}, Z, D=0]\Big) \\
    & \hspace{25pt} + \E[\Delta Y_{t^*}(0) | X_{t^*}(0), X_{t^*-1}, Z, D=0] \\
    &= ATT_{X^{all}(0)}(X^{all}(0)) + 0 + \E[\Delta Y_{t^*}(0) | X_{t^*}(0), X_{t^*-1}, D=0] \\
    &= \E[Y_{t^*}(1) - Y_{t^*}(0) | X_{t^*}(0), X_{t^*-1}, D=1] + \E[\Delta Y_{t^*}(0) | X_{t^*}(0), X_{t^*-1}, D=0] \\
    &= \Big(\E[\Delta Y_{t^*} | X_{t^*}, X_{t^*-1}, D=1] - \E[\Delta Y_{t^*}(0) | X_{t^*}(0), X_{t^*-1}, D=1]\Big) \\
    & \hspace{25pt} + \E[\Delta Y_{t^*}(0) | X_{t^*}(0), X_{t^*-1}, D=0] \\
    &= \E[\Delta Y_{t^*} | X_{t^*}, X_{t^*-1}, D=1]
\end{align*}
where the first equality holds by adding and subtracting terms, the second equality holds by \Cref{ass:cov-exog} and the definition of $ATT_{X^{all}(0)}$, by \Cref{ass:conditional-parallel-trends}, and by \Cref{ass:bias-z}(b), the third equality holds by \Cref{ass:bias-z}(a), the fourth equality holds by adding and subtracting terms and by \Cref{ass:cov-exog}, and the last equality holds because
\begin{align*}
    \E[\Delta Y_{t^*}(0) | X_{t^*}(0), X_{t^*-1},D=1] &= \E[ \E[\Delta Y_{t^*}(0) | X_{t^*}(0), X_{t^*-1}, Z, D=1] | X_{t^*}(0), X_{t^*-1},D=1] \\
    &= \E[ \E[\Delta Y_{t^*}(0) | X_{t^*}(0), X_{t^*-1}, Z, D=0] | X_{t^*}(0), X_{t^*-1},D=1] \\
    &= \E[ \E[\Delta Y_{t^*}(0) | X_{t^*}(0), X_{t^*-1}, D=0] | X_{t^*}(0), X_{t^*-1},D=1] \\
    &= \E[\Delta Y_{t^*}(0) | X_{t^*}(0), X_{t^*-1}, D=0]
\end{align*}
where the first equality holds by the law of iterated expectations, the second equality holds by \Cref{ass:conditional-parallel-trends}, the third equality holds by \Cref{ass:bias-z}(b), and the last equality holds because, conditional on $X_{t^*}(0)$, and $X_{t^*-1}$, the inside conditional expectation is non-random.  Thus, Term (B) is equal to zero under \Cref{ass:cov-exog} and \Cref{ass:bias-z}.

\paragraph{Term (C):}  First, consider the expression in Term (C) for $d=0$.  Notice that,
\begin{align*}
    \E[\Delta Y_{t^*} | X_{t^*}, X_{t^*-1}, D=0] &= \E[\Delta Y_{t^*}(0) | X_{t^*}(0), X_{t^*-1}, D=0] \\
    &= \E[\Delta Y_{t^*}(0) | \Delta X_{t^*}(0), D=0] \\
    &= \E[\Delta Y_{t^*} | \Delta X_{t^*}, D=0]
\end{align*}
where the first equality holds by replacing observed counterparts with their corresponding potential outcomes, the second equality holds by \Cref{ass:bias-delta}(b), and the third equality holds by writing potential outcomes in terms of their observed counterparts.  Second, consider the case when $d=1$, so that
\begin{align*}
    \E[\Delta Y_{t^*} | X_{t^*}, X_{t^*-1}, D=1] &= \Big(\E[\Delta Y_{t^*} | X_{t^*}, X_{t^*-1}, D=1] - \E[\Delta Y_{t^*}(0) | X_{t^*}(0), X_{t^*-1}, D=1] \Big) \\
    &\hspace{25pt} + \Big( \E[\Delta Y_{t^*}(0) | X_{t^*}(0), X_{t^*-1}, D=1] - \E[\Delta Y_{t^*}(0) | X_{t^*}(0), X_{t^*-1}, D=0]\Big) \\
    & \hspace{25pt} + \E[\Delta Y_{t^*}(0) | X_{t^*}(0), X_{t^*-1}, D=0] \\
    &= ATT_{X_{t^*}(0), X_{t^*-1}}(X_{t^*}(0), X_{t^*-1}) + 0 + \E[\Delta Y_{t^*}(0) | \Delta X_{t^*}(0), D=0] \\
    &= \E[Y_{t^*}(1) - Y_{t^*}(0) | \Delta X_{t^*}(0), D=1] + \E[\Delta Y_{t^*}(0) | \Delta X_{t^*}(0), D=0] \\
    &= \Big(\E[\Delta Y_{t^*} | \Delta X_{t^*}, D=1] - \E[\Delta Y_{t^*}(0) | \Delta X_{t^*}(0), D=1]\Big) \\
    &\hspace{25pt} + \E[\Delta Y_{t^*}(0) | \Delta X_{t^*}(0), D=0] \\
    &=\E[\Delta Y_{t^*} | \Delta X_{t^*}, D=1]
\end{align*}
where the first equality holds by adding and subtracting terms, the second equality holds by \Cref{ass:cov-exog}, the definition of $ATT_{X_{t^*}(0),X_{t^*-1}}$, \Cref{ass:bias-delta}(b), and the middle term is equal to zero by the same arguments as were used for Term (B) above, the third equality holds by \Cref{ass:bias-delta}(a), the fourth equality holds by adding and subtracting terms and by \Cref{ass:cov-exog}, and the last equality holds because
\begin{align*}
    \E[\Delta Y_{t^*}(0) | \Delta X_{t^*}(0), D=1] &= \E[ \E[\Delta Y_{t^*}(0) | X_{t^*}(0), X_{t^*-1}, Z, D=1] | \Delta X_{t^*}(0), D=1] \\
    &= \E[ \E[\Delta Y_{t^*}(0) | X_{t^*}(0), X_{t^*-1}, Z, D=0] | \Delta X_{t^*}(0), D=1] \\
    &= \E[ \E[\Delta Y_{t^*}(0) | \Delta X_{t^*}(0), D=0] | \Delta X_{t^*}(0), D=1] \\
    &= \E[\Delta Y_{t^*}(0) | \Delta X_{t^*}(0), D=0]
\end{align*}
where the first equality holds by the law of iterated expectations, the second equality holds by \Cref{ass:conditional-parallel-trends}, the third equality holds by \Cref{ass:bias-z}(b) and \Cref{ass:bias-delta}(b), and the last equality holds because, conditional on $\Delta X_{t^*}(0)$, the inside conditional expectation is nonrandom.  Thus, Term (C) is equal to zero under \Cref{ass:cov-exog,ass:bias-z,ass:bias-delta}.

\paragraph{Term (D):} First, consider the expression in Term (D) for $d=0$.  Notice that
\begin{align*}
    \E[\Delta Y_{t^*} | \Delta X_{t^*}, D=0] &= \E[\Delta Y_{t^*}(0) | \Delta X_{t^*}(0), D=0] \\
    &= \Delta X_{t^*}(0)'\delta_0 \\
    &= \L(\Delta Y_{t^*} | \Delta X_{t^*}, D=0) 
\end{align*}
where the first equality holds by writing observed outcomes and time-varying covariates in terms of potential outcomes/covariates, the second equality holds by \Cref{ass:bias-linearity}(b), and the last equality holds by the definition of linear projection.  Next, consider the expression in Term (D) for $d=1$,
\begin{align*}
    \E[\Delta Y_{t^*} | \Delta X_{t^*}, D=1] &= \Big(\E[\Delta Y_{t^*} | \Delta X_{t^*}, D=1] - \E[\Delta Y_{t^*}(0) | \Delta X_{t^*}(0), D=1] \Big) + \E[\Delta Y_{t^*}(0) | \Delta X_{t^*}(0), D=1] \\
    &= ATT_{\Delta X_{t^*}(0)}(\Delta X_{t^*}(0)) + \E[\Delta Y_{t^*}(0) | \Delta X_{t^*}(0), D=1] \\
    &= \Delta X_{t^*}(0)'(\delta_1 + \delta_0) \\
    &= \L(\Delta Y_{t^*} | \Delta X_{t^*}, D=1)
\end{align*}
where the first equality holds by adding and subtracting terms, the second equality holds using similar arguments as for previous terms and uses \Cref{ass:conditional-parallel-trends,ass:cov-exog,ass:bias-z,ass:bias-delta}, the third equality holds by \Cref{ass:bias-linearity}, and the last equality holds by the definition of linear projection where the linear projection coefficient is given by $\delta_1+\delta_0$.

This completes the first part of the proof.  Next, we prove additional result (a) in \Cref{prop:twfe-extra-assumptions}.  Toward this end, recall that,
\begin{align*}
    \omega_e(X^{all}(0)) = \frac{ (p(X^{all}(0)) - \L(D|\Delta X))}{\E[(1-\L(D|\Delta X_{t^*}))|D=1]}
\end{align*}
Under \Cref{ass:bias-pscore}, $p(X^{all}(0)) = \L(D|\Delta X_{t^*})$ which implies that $\omega_e(X^{all}(0)) = 0$.  Next, recall that
\begin{align*}
    \omega_{ATT}(X^{all}(0)) &= \frac{1-p(X^{all}(0))}{\E[(1-\L(D|\Delta X))|D=1]} \\
    &= \frac{1-p(X^{all}(0))}{\E[(1-p(X^{all}(0))|D=1]}
\end{align*}
where the second equality holds under \Cref{ass:bias-pscore}.  It immediately follows that $\E[\omega_{ATT}(X^{all}(0))|D=1]=1$
This completes the proof of additional result (a).  Now, we move to proving additional result (b) in \Cref{prop:twfe-extra-assumptions}.  Under \Cref{ass:sampling,ass:conditional-parallel-trends,ass:overlap,ass:cov-exog,ass:bias-z,ass:bias-delta,ass:bias-linearity,ass:bias-pscore}, we have shown that 
\begin{align*}
    \alpha &= \E[\omega_{ATT}(X^{all}(0)) ATT_{\Delta X_{t^*}(0)}(\Delta X_{t^*}(0)) | D=1] \\
    &= ATT \, \E[\omega_{ATT}(X^{all}(0)) | D=1] \\
    &= ATT
\end{align*}
where the second equality holds when $ATT_{\Delta X_{t^*}(0)}$ does not vary across different values of $\Delta X_{t^*}(0)$, and the last equality holds because the weights have mean one.

\begin{additional-material}
\section{Notes}

\begin{itemize}
    \item Would be nice to have a more concise notation that includes time varying and time invariant covariates.  Think about this, but one possibility is $Z=(X_{t-1}, \tilde{X})$ where $\tilde{X}$ is a vector of time invariant covariates.
    \item If more than two time periods are available, pre-test the parallel trends assumption and the assumptions about the paths of covariates.
    \item Defining conditional treatment effects is trickier here than in the case with time invariant covariates only. In the main case in the current paper, one could consider $ATT_0(x_0) = E[Y_t(1) - Y_t(0)|X(0)=x_0, D=1]$ or $ATT_0(x_1) = E[Y_t(1) - Y_t(0)|X(1)=x_1, D=1]$.
    \item \Cref{ass:cov-unc} implies that ATT is identified, it is not strong enough to identify $ATT_X(x)$. It is strong enough to identify $ATT^0(x_0)$ but this is not an ATT for an observable group in the population.
    \item It might be interesting to find an application where we could nonparametrically estimate a lot of these conditional expectations and compute how much bias arises from each term.
    \item the identification challenge is for 
    \begin{align*}
        E[Y_t(0)|D=1] = E[\Delta Y_t(0) | D=1] + E[Y_{t-1}|D=1]
    \end{align*}
    The last term is directly identified, so we focus on the first term on the RHS:
    \begin{align*}
        E[\Delta Y_t(0) | D=1] = (\theta_t - \theta_{t-1}) + E[\Delta X_t(0) | D=1]\beta
    \end{align*}
    A sufficient condition for this procedure to deliver the ATT is that $E[\Delta X_t(0) | D=1] = E[\Delta X_t(0) | D=0]$ (but do you get the same thing as in the regression case here...?).  Under this assumption,
    \begin{align*}
        ATT = E[\Delta Y_t|D=1] - E[\Delta Y_t|D=0]
    \end{align*}
    so you get the same thing, that is interesting..., it does suggest that you do not actually need to run the regression, and here we do not put any structure on how treated potential outcomes are generated (unlike TWFE)
    \item maybe put some quotes from papers about why they condition on covariates... (e.g., from cheng and hoekstra or others)
    \item one of pedro's comments was about whether we actually needed the functional form assumption for $\Delta X_t(0)$ in that case for the result to go through.  I think the answer is yes.  You can show this by use nonparametric arguments and then assuming something like $\E[\Delta X_t(0)|Z,X_{t-1},D=0] = Z^2\lambda$ which would imply that you can't just run a regression on $Z$ and $X_{t-1}$ due to the nonlinearity.
\end{itemize}

\subsection*{TODO}

Besides the long reading list below, here is a recap of what I think we need to be working on:

\begin{itemize}
    \item Work on Corollary 2.  In particular, we need a proof of the first part (which I think should be straightforward).  For the second part, we need to find something that is doubly robust, and make a connection to existing work on double robustness under sequential exogeneity.
    \item Work on interpreting regressions.  I was quickly scanning through some of the papers mentioned below, and I think they suggest that we are on the right track there (I think the Angrist paper is a famous reference, but the Aronow paper and the Goldsmith-Pinkham papers both have easy-to-follow and relevant discussions.  There are still some issues related to confirming that weights are exactly right and interpreting them though.
    \item Machine learning --- I'd like to go for this, but it probably involves better figuring out the expressions in Corollary 2, which might be a nightmare to try to estimate as currently constructed.
    \item Application.  As I said earlier, I think we should defer a bit on this one, but we'll eventually need a lot of work here too.
\end{itemize}

\subsection*{Reading List}

This paper studies a similar problem to what we do (though there are still important differences, we should cite as being similar):

\begin{itemize}
    \item \fullcite{zeldow-hatfield-2021} 
    
    **NOTE: I could not find the Appendix, so this review may be incomplete.**
    
    This paper examines confounding variables and looks at common regression and matching methods for dealing with confounding variables in a DiD context.
    
    The paper says that a confounding variable in DiD is a variable that either has a time-varying effect on the outcome or has a time-varying difference between the treatment and control groups.
    
    For time-invariant confounders, the variable must thus have a time-varying effect. As such, simply controlling for the main effect of the confounding variable is insufficient, and researchers should instead also control for the interaction between the variable and time.
    
    For time-varying confounders, conditioning on the variable's main effect or on its interaction with time runs the risk of conditioning on post-treatment covariates that are impacted by the treatment. However, failing to control for the covariate may lead to violations of the parallel trends assumption.
    
    Matching on pre-treatment outcomes can reduce bias in some scenarios and increase bias in others by inducing regression to the mean. Matching on pre-treatment covariates is also able to control for diverging trends, but this would not work if the confounding was due to time-varying covariates.
    
    The paper includes several simulations to test the ability of six common regression and matching methods to deal with different confounding scenarios. The models are the following:
    
    1. A simple, unadjusted DiD estimator.
    
    2. A covariate regression adjusted DiD estimator.
    
    3. A time-varying covariate regression adjusted DiD estimator.
    
    4. A DiD estimator that matches based on pretreatment outcomes.
    
    5. A DiD estimator that matches based on pretreatment outcome first differences.
    
    6. A DiD estimator that matches based on pretreatment covariates.
    
    Note that all of the matching estimators used nearest neighbor matching.
    
    For a time-invariant covariate effect (which has no confounding), none of the estimators had significant bias. For a time-varying covariate effect, only models 3 and 6 had essentially no bias. For a treatment-independent covariate (which has no confounding), all of the models were unbiased except for model 4, and model 3 had about 10\% lower variance than the simple, unadjusted model.
    
    For a time-varying covariate with parallel evolution and constant effect of X on Y (which has no confounding), all methods were unbiased. However, for a time-varying covariate with parallel evolution but a changing effect of X on Y, only model 3 was unbiased.
    
    For a time-varying covariate with diverging evolution by group but the effect of X on Y is constant, models 2 and 3 are unbiased. For a time-varying covariate with diverging evolution by group and the effect of X on Y is changing, only model 3 is unbiased.
    
    For an evolution that diverges after the treatment is introduced, all methods are biased, regardless of whether the effect of X on Y is constant.
    
    Notably, matching estimators never perform than regression adjusted estimators (although it is important to remember there are other matching methods besides nearest neighbor).

    \point{Mostly harmless talks about "bad controls" (though only very briefly in the context of panel data), which is very closely related to what we are talking about here, but they don't appear to include any references.  In addition, the example of a bad control that is used in mostly harmless is controlling for occupation in a regression of earnings on education.}
\end{itemize}

\paragraph{Sequential Exogeneity / Mediation Analysis}

This is very incomplete list; feel free to add papers that seem relevant

\begin{itemize}
    \item Lechner, Michael. "Sequential causal models for the evaluation of labor market programs." Journal of Business and Economic Statistics 27.1 (2009): 71-83.
    
    This article develops estimators that build off of Lechner and Miquel (2001) for the model proposed by Robins (1986). The goal of these estimators is to overcome selection issues that may plague the static model when dealing with situations such as treatments in multiple time periods or continuous treatments. The article uses IPW estimation of dynamic causal models. One advantage of this model is that it does not necessitate modeling the dependence of the outcomes and the confounders.
    
    \item \fullcite{lechner-2008}
    \item \fullcite{flores-lagunes-2009}
    \item \fullcite{rosenbaum-1984}
    
    This paper examines when it is appropriate to control for concomitant variables that may have been adjusted by the treatment to control for confounding variables that may be unobserved.
    
    Define $R_t$ and $S_t$ to be the outcome variable and concomitant variable, respectively, where t = 0 if untreated and t = 1 if treated. Let X be observed covariates. Let Z be a binary variable that equals 1 if the unit was treated and 0 otherwise.
    
    Consider the following "strongly ignorable" condition:
    
    \begin{equation}
        (R_1, R_0) \perp Z | X
    \end{equation}
    
    and
    
    \begin{equation}
        0 < Pr(Z = 1 | X) < 1 for all X
    \end{equation}
    
    If this is true, then
    
    \begin{equation}
        \bar{\delta} = E_X{E(R_1 | Z = 1, X) - E(R_0 | Z = 0, X)}
    \end{equation}
    
    is equal to the average treatment effect. Adjusting for concomitant variables can introduce bias, but there are certain cases in which it does not. Notably, adjusting for concomitant variables does not introduce bias when the concomitant variables are unaffected by treatment and treatment assignment is strongly ignorable given X. However, this also implies that it is sufficient to control for X to recover the average treatment effect, making controlling for $S_t$ only unbiased when it is unnecessary.
    
    Controlling for $S_t$ may be worthwhile if $S_t$ can act for a surrogate of an unobserved confounding variable U. $(S_0, S_1)$ can be a surrogate of U if
    
    \begin{equation}
        R_t \perp U | (S_t, X) \text{for t= 0, 1}
    \end{equation}
    
    Other assumptions include that treatment assignment is strongly ignorable for ${(R_1, S_1), (R_0, S_0)}$ given (X, U) and that $(S_1, S_0)$ are unaffected by the treatment. These assumptions together are quite strong but may make sense in certain scenarios.
    
    One alternative method is to take the difference of the average treatment effect calculated after controlling only for X and the average treatment effect after controlling for X and $S_t$. This implicitly assumes that the true average treatment effect is between the two, but this is only guaranteed to be true with additional assumptions. However, some researchers still use this in analyses of how sensitive conclusions are to assumptions. Using data from previous studies and/or logical arguments may yield a possible plausible range of average treatment effect values.
    
    \item Viviano, Davide, and Jelena Bradic. "Dynamic covariate balancing: estimating treatment effects over time." arXiv preprint arXiv:2103.01280 (2021). (includes references to double robustness)
    \item (some paper that does double robustness in this context)
\end{itemize}

\paragraph{Interpreting Regressions}

\begin{itemize}
    \item Angrist, Joshua D. "Estimating the Labor Market Impact of Voluntary Military Service Using Social Security Data on Military Applicants." Econometrica 66.2 (1998): 249-288.
    \item \fullcite{sloczynski-2020}
    
    This paper looks at interpreting weighting for the OLS and TSLS estimands. Notably, as more units are treated, less weight is placed on the effect of the treated. This paper imposes a linearity assumption on propensity scores, which allows the weights to be mean 1.
    
    The bias from imposing linearity can be written as:
    
    \begin{equation}
        B_l = w_0 * (\tau_{APE|d=0} - \tau_{ATC}) + w_1 * (\tau_{APE|d=1} - \tau_{ATT})
    \end{equation}
    
    The bias from heterogeneity can be written as:
    
    \begin{equation}
        B_h = w_0 * (\tau_{ATC} - \tau_{ATT})
    \end{equation}
    
    The paper focuses mainly on correcting the bias from heterogeneity, but Sloczynski argues that the bias from linearity is relatively unimportant in the applications he is considering, while conceding it may be more substantial in other empirical applications.
    
    The paper proposes a weighted least squares regression. First make the following assumptions:
    
    Assumption 1: (i) $E(y^2)$ and $E(||X||^2)$ are finite, (ii) The covariance matrices of X and (d, X) are nonsingular.
    
    Assumption 2: V[p(X) | d = 1] and V[p(X) | d = 0] are nonzero, where V(-|-) denotes the conditional variance (with respect to E[p(X) | d = j], j = 0, 1).
    
    Also, let $w = \frac{1-\rho}{w_0} * d + \frac{\rho}{w_1}*(1 - d)$. Then,
    
    \begin{equation}
        \tau_w = \tau_{APE}
    \end{equation}
    
    The paper shows that the OLS and TSLS estimands are equivalent to convex combinations of two other parameters. These parameters can be interpreted as ATT and ATC.
    
    To estimate ATE, the proportion of treated units should be similar to the proportion of control units. To estimate ATT, the proportion of treated units should be very small.
    
    \item Aronow, Peter M., and Cyrus Samii. "Does regression produce representative estimates of causal effects?." American Journal of Political Science 60.1 (2016): 250-267.
    
    This paper argues that regressions do not produce externally valid results because observations are weighted differently. This results in an effective sample that does not match the population.
    
    The authors make the following assumptions, which they say are favorable to regressions:
    \begin{equation}
        Y_i(d) = Y_i(0) + \tau_i * d
    \end{equation}
    
    \begin{equation}
        (Y_i(0), \tau_i) \perp D_i | X_i
    \end{equation}
    
    \begin{equation}
        E[D_i | X_i] = V_i\omega],
    \end{equation}
    
    where $V_i = (1 X_i)$ and $\omega$ is a K + 1 column vector of coefficients.
    
    The authors find that the weighted average of causal effects from partial regression is the following:
    
    \begin{equation}
        \hat{\beta} \overset{p}{\to} \frac{E[w_i\tau_i]}{E[\tau_i]}
    \end{equation}
    
    where $w_i = (D_i - E[D_i | X_i])^2$
    
    In addition, using the definition of conditional variance,
    
    \begin{equation}
        E[w_i | X_i] = Var[D_i | X_i]
    \end{equation}
    
    To estimate the mean of a covariate Z in the effective sample,
    
    \begin{equation}
        \hat{\mu}(Z_i) = \frac{\sum_{i=1}^n\hat{w_i}Z_i}{\sum_{i=1}^n \hat{w_i}} \overset{p}{\to} \frac{E[w_iZ_i]}{E[w_i]} = \mu(Z_i)
    \end{equation}
    
    Researchers may be able to uncover the causal effect in the target population if the following assumptions hold:
    
    \begin{equation}
        (Y_i(d), Y_i(d')) \perp D_i | X_i
    \end{equation}
    
    \begin{equation}
        Pr[D_i = d | X__i = x] > 0, Pr[D_i = d' | X_i = x] > 0
    \end{equation}
    
    If these hold, than the average causal effect for the target population is:
    
    \begin{equation}
        \bar{\tau}(d, d') = E[\tau_i(d, d')] = \frac{E[Y_i(d') - E[Y_i(d)]]}{d' - d}
    \end{equation}
    
    \item Humphreys, Macartan. "Bounds on least squares estimates of causal effects in the presence of heterogeneous assignment probabilities." Manuscript, Columbia University (2009).
    \item Goldsmith-Pinkham, Paul, Peter Hull, and Michal Kolesár. "On Estimating Multiple Treatment Effects with Regression." arXiv preprint arXiv:2106.05024 (2021).
    
    \item Ishimaru, Shoya. "Empirical Decomposition of the IV-OLS Gap with Heterogeneous and Nonlinear Effects." arXiv preprint arXiv:2101.04346 (2021).
    
    \item Ishimaru, Shoya. "What Do We Get From A Two-Way Fixed Effects Estimator? Implications From A General Numerical Equivalence." arXiv preprint arXiv:2103.12374 (2021).
\end{itemize}

\paragraph{Machine Learning}

\begin{itemize}
    \item \fullcite{chang-2020}
    \item Victor Chernozhukov, Denis Chetverikov, Mert Demirer, Esther Duflo, Christian Hansen, Whitney Newey, James Robins, Double/debiased machine learning for treatment and structural parameters, The Econometrics Journal, Volume 21, Issue 1, 1 February 2018, Pages C1–C68
    \item Semenova, Vira, and Victor Chernozhukov. "Debiased machine learning of conditional average treatment effects and other causal functions." The Econometrics Journal 24.2 (2021): 264-289.
    \item Colangelo, Kyle, and Ying-Ying Lee. "Double debiased machine learning nonparametric inference with continuous treatments." arXiv preprint arXiv:2004.03036 (2020).
    \item Su, Liangjun, Takuya Ura, and Yichong Zhang. "Non-separable models with high-dimensional data." Journal of Econometrics 212.2 (2019): 646-677.
    \item Chernozhukov (2018) — this paper is discussed a lot in Chang (2020).
\end{itemize}

\end{additional-material}

\end{document}